%% file: main.tex
\documentclass[10pt,twocolumn]{article}

\usepackage{booktabs} % For formal tables
\usepackage{multicol}
\usepackage{todonotes}
\usepackage{amsfonts} 
\usepackage{amsmath} 
\usepackage{amsthm}
\usepackage{stmaryrd}
\usepackage{listings}
\usepackage{xspace}
\usepackage{cleveref}
\newtheorem{definition}{Definition}
\newtheorem{theorem}{Theorem}
\newtheorem{lemma}{Lemma}
\newtheorem{proposition}{Proposition}

\newtheorem{example}{Example}
\newtheorem{corrolary}{Corrolary}
\usepackage[noend]{algorithmic}
\usepackage{algorithm}
\usepackage{enumitem}
\usepackage{multirow}
\usepackage{pifont}
\usepackage{subcaption}
\usepackage{url}

\newif\ifshowallproofs
\showallproofsfalse
\newcommand{\HiddenProof}[1]{
  \ifshowallproofs
  #1
  \fi
}
 
\newcommand{\IndexName}{\texttt{JITD}\xspace}
\newcommand{\IndexNames}{\texttt{JITD}s\xspace}
\newcommand{\CogName}{\textsc{cog}\xspace}
\newcommand{\ArrayAtom}{\textbf{Array}}
\newcommand{\SortedAtom}{\textbf{Sorted}}
\newcommand{\ConcatAtom}{\textbf{Concat}}
\newcommand{\BTreeAtom}{\textbf{BinTree}}
\newcommand{\SortTransform}{\textbf{Sort}}
\newcommand{\UnsortTransform}{\textbf{UnSort}}
\newcommand{\MergeTransform}{\textbf{Merge}}
\newcommand{\DivideTransform}{\textbf{Divide}}
\newcommand{\CrackTransform}{\textbf{Crack}}
\newcommand{\PivotLeftTransform}{\textbf{PivotLeft}}
\newcommand{\PivotRightTransform}{\textbf{PivotRight}}
\newcommand{\IdentityTransform}{\textbf{id}}
\newcommand{\SortFunction}{\texttt{sort}}

\newcommand{\AtomicTransformSet}{\mathcal A}
\newcommand{\HierarchicalTransforms}{\Delta}
\newcommand{\LHSMetaTransform}{\textbf{LHS}}
\newcommand{\RHSMetaTransform}{\textbf{RHS}}
\newcommand{\PolicyInstance}{\mathcal P}
\newcommand{\PolicyDomain}{\mathcal D}
\newcommand{\PolicyHeuristic}{\mathcal H}
\newcommand{\ScoreFunction}{\texttt{score}}
\newcommand{\WeightedTargets}{\mathcal W}

\newcommand{\RootName}{\texttt{root}}

\newcommand{\TypeOfCogInstance}[1]{\textbf{typeof}(#1)}
\newcommand{\GenericType}{\tau}
\newcommand{\CogType}{\mathcal C}
\newcommand{\KeyType}{\mathcal I}
\newcommand{\TransformType}{\mathcal T}

\newcommand{\RecordType}{\mathcal R}
\newcommand{\ArrayOfType}[1]{\left[ #1 \right]}
\newcommand{\BagOfType}[1]{\left \{\hspace*{-0.5mm}\left|\hspace{0.5mm} #1 \hspace{0.5mm}\right|\hspace*{-0.5mm}\right\}}
\newcommand{\SetOfType}[1]{\left \{ #1 \right \}}

\newcommand{\KeyInstance}{k}
\newcommand{\RecordInstance}{r}
\newcommand{\TransformInstance}{T}
\newcommand{\MetaTransformInstance}{M}

\newcommand{\ArrayOf}[1]{\left[ #1 \right]}
\newcommand{\BagOf}[1]{\left \{\hspace*{-0.5mm}\left|\hspace{0.5mm} #1 \hspace{0.5mm}\right|\hspace*{-0.5mm}\right\}}
\newcommand{\SetOf}[1]{\left \{ #1 \right \}}
\newcommand{\Wildcard}{\;\_\;}
\newcommand{\Comprehension}[2]{\left\{\;{#1}\;|\;{#2}\;\right\}}
\newcommand{\Tuple}[1]{\left<\;{#1}\;\right>}

\newcommand{\KeyOrder}{\preceq}
\newcommand{\KeyOrderStrict}{\prec}
\newcommand{\KeyOrderRev}{\succeq}
\newcommand{\KeyOf}[1]{\texttt{id}\left(#1\right)}

\newcommand{\DescendantsOfCog}[1]{#1^*}
\newcommand{\ContentsOfCog}[1]{\mathbb{D}\left(#1\right)}
\newcommand{\IsStructurallyCorrect}[1]{\textsc{StrCor}\left(#1\right)}
\newcommand{\IsLogicallyEquivalentTo}{\approx}
\newcommand{\FloorOf}[1]{\left \lfloor #1 \right \rfloor}
\newcommand{\IsDefinedAs}{\;\overset{def}{=}\;}
\newcommand{\TransformIsCorrect}[1]{\textsc{Correct}\left(#1\right)}
\newcommand{\TraceOf}{\textbf{Trace}}

\newcommand{\tinysection}[1]{\medskip \noindent \textbf{#1}. }
\newcommand{\trimfigurespacing}{\vspace*{-5mm}}

\begin{document}
%\begin{multicols}{2}

% \title{``Just-in-Time'' Index Transitions} 
\title{Just-in-Time Index Compilation}

\author{
Darshana Balakrishnan, Lukasz Ziarek, Oliver Kennedy\\
\textbf{University at Buffalo}\\
\texttt{\{\;dbalakri, lziarek, okennedy\;\}@buffalo.edu}
}
\date{}

\maketitle

\begin{abstract}
\input{sections/abstract}
\end{abstract}

\section{Introduction}
\label{sec:introduction}
\input{sections/introduction}

\subsection{System Overview}
\label{sec:overview}
\input{sections/overview.tex}

\section{A Grammar of Data Structures}
\label{sec:grammar}
\input{sections/grammar}

\section{Transforms over \CogName}
\label{sec:transformations}
\input{sections/transformations}

\section{Policies for Transforms}
\label{sec:policies}
\input{sections/policies}

\section{Implementing The \IndexName Runtime}
\label{sec:implementation}
\input{sections/implementation}

\section{Policy Optimization}
\label{sec:model}
\input{sections/model}

\section{On the Generality of \IndexNames}
\label{sec:discussion}
\input{sections/discussion}

\section{Evaluation}
\label{sec:evaluation}
\input{sections/evaluation}

\section{Related Work}
\label{sec:rel_wrk}
\input{sections/related_work}

\section{Conclusions and Future Work}
\label{sec:fut_wrk}
\input{sections/future_work}

% \section{Acknowledgements}
% This work is supported by NSF awards IIS-1617586 and IIS-1750460.  
% The conclusions and opinions in this work are solely those of the authors and do not represent the views of the National Science Foundation.
% The authors would like to thank Hank Lin for substantial contributions to a preliminary implementation of just-in-time indexes.

\bibliographystyle{abbrv}
\bibliography{main,oliver}

\appendix
\input{sections/proofs.tex}

\end{document}

%% file: sections/abstract.tex
% -*- root: ../main.tex -*-
Creating or modifying a primary index is a time-consuming process, as the index typically needs to be rebuilt from scratch.
In this paper, we explore a more graceful ``just-in-time'' approach to index reorganization, where small changes are dynamically applied in the background.
To enable this type of reorganization, we formalize a composable organizational grammar, expressive enough to capture instances of not only existing index structures, but arbitrary hybrids as well.
We introduce an algebra of rewrite rules for such structures, and a framework for defining and optimizing policies for just-in-time rewriting.
Our experimental analysis shows that the resulting index structure is flexible enough to adapt to a variety of performance goals, while also remaining competitive with existing structures like the C++ standard template library map.

%% file: sections/introduction.tex
% -*- root: ../main.tex -*-

An in-memory index is backed by a data structure that stores and facilitates access to records.
An alphabet soup of such data structures have been developed to date (\cite{DBLP:conf/focs/GuibasS78,DBLP:journals/csur/Comer79,DBLP:conf/dagstuhl/Graefe05,DBLP:conf/iccsa/LeeBB09,DBLP:journals/puc/ZhouTZNW14,DBLP:journals/cacm/Larson88,DBLP:journals/acta/ONeilCGO96,DBLP:conf/sigmod/SearsR12,okasaki1996functional,DBLP:conf/cidr/IdreosKM07} to list only a few).
Each structure targets a specific trade-off between a range of performance metrics (e.g., read cost, write cost), resource constraints (e.g., memory, cache), and supported functionality (e.g. range scans or out-of-core storage).
As a trivial example, contrast linked lists with a sorted arrays: 
The former provides fast writes and slow lookups, while the latter does exactly the opposite.

\begin{figure}
\centering
\includegraphics[width=0.8\columnwidth]{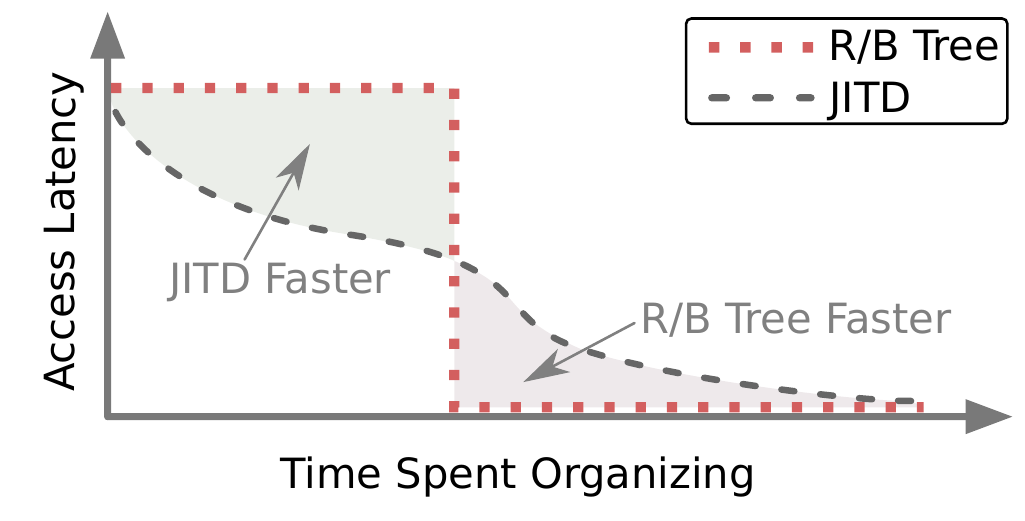}
\caption{A classical index data structure provides no benefits until ready, while \IndexNames provide continuous incremental performance improvements.}
\label{fig:scalingComparison}
\trimfigurespacing
\end{figure}

Creating or modifying an in-memory index is a time-consuming process, since the data structure backing the index typically needs to be rebuilt from scratch when its parameters change.
During this time, the index is unusable,  penalizing the performance of any database relying on it.
In this paper, we propose a more graceful approach to runtime index adaptation.
Just-in-Time Indexes (\IndexNames) continuously make small, incremental reorganizations in the background, while client threads continue to access the structure.
Each reorganization brings the \IndexName closer to a state that mimicks a specific target data structure.
As illustrated in Figure~\ref{fig:scalingComparison}, the performance of a \IndexName continuously improves as it transitions from one state to another, while other data structures improve only after fixed investments of organizational effort.

Three core challenges must be addressed to realize \IndexNames.
First, because each individual step is small, at any given point in time an \IndexName may need to be in some intermediate state between two classical data structures.
For example, an \IndexName transitioning from a linked list to a binary tree may need to occupy a state that is neither linked list, nor binary tree, but some combination of the two.
Second, there may be multiple pathways to transition from a given source state to the desired target state.
For example, to get from an unsorted array to a sorted array, we might sort the array (faster in the long-term) or crack~\cite{DBLP:conf/cidr/IdreosKM07} the array (more short-term benefits).  
Finally, we want to avoid blocking client access to the \IndexName while it is being reorganized.
Client threads should be able to query the structure while the background thread works.

We address the first challenge by building on prior work with \IndexNames~\cite{DBLP:conf/cidr/KennedyZ15}, where we defined them as a form of adaptive index that dynamically assembles indexes from composable, immutable building blocks.
Mimicking the behavior of a just-in-time compiler, a just-in-time data structure dynamically reorganizes building blocks to improve index performance.
Our main contributions in this paper address the remaining challenges.

We first precisely characterize the space of available state transitions by formalizing the behavior of \IndexNames into a composable organizational grammar (\CogName).
A sentence in \CogName corresponds directly to a specific physical layout.
Many classical data structures like binary trees, linked lists, and arrays are simply syntactic restrictions on \CogName.
Lifting these restrictions allows intermediate hybrid structures that combine elements of each.
Thus, the grammar can precisely characterize any possible state of a \IndexName.

Next, we define \emph{transforms}, syntactic rewrite rules over \CogName and show how these rewrite rules can be combined into a \emph{policy} that dictate how and where transforms should be applied.
This choice generally requires runtime decisions, so we identify a specific family of ``local hierarchical'' policies in which runtime decisions can be implemented by an efficiently maintainable priority heap.
As an example, we define a family of policies for transitioning between unsorted and sorted arrays (e.g., for interactive analysis on a data file that has just been loaded~\cite{DBLP:conf/sigmod/AlagiannisBBIA12}).

To automate policy design, we provide a simulator framework that predictively models the performance of a \IndexName under a given policy.  
The simulator can generate performance-over-time curves for a set of potential policies.
These curves can then be queried to find a policy that best satisfies user desiderata like ``get to $300ms$ lookups as soon as possible'' or ``give me the best scan performance possible within $5s$''.

Finally, we address the issue of concurrency by proposing a new form of ``semi-functional'' data structure.  
Like a functional (immutable) data structure, elements of a semi-functional data structure are stable once created.  
However, using handle-style~\cite{handles} pointer indirection, we draw a clear distinction between code that expects physical stability and code that merely expects logical stability.
In the latter case correctness is preserved even if the element is modified, so long as the element's logical content remains unchanged.

% \noindent Concretely, the contributions of this paper include:
% % \begin{enumerate}
% (1) We formalize the just-in-time data structures of \cite{DBLP:conf/cidr/KennedyZ15} into a composable organizational grammar for describing physical layouts.
% (2) We define and prove the correctness of multiple equivalence-preserving rewrite rules called transforms.
% (3) We define transform-based policies for transitioning between different structures and identify a class of local hierarchical policies that can be implemented with minimal management overhead.
% (4) We describe the \IndexName runtime, which includes support for local hierarchical policies and background worker threads for transitioning indexes.
% (5) We describe the implementation of a \IndexName index simulator that aids in optimizing policy trade-offs.
% (6) We implement and evaluate a policy for organizing newly loaded data.
% % \end{enumerate}

%% file: sections/overview.tex
% -*- root: ../main.tex -*-

\begin{figure}
\centering
\includegraphics[width=0.8\columnwidth]{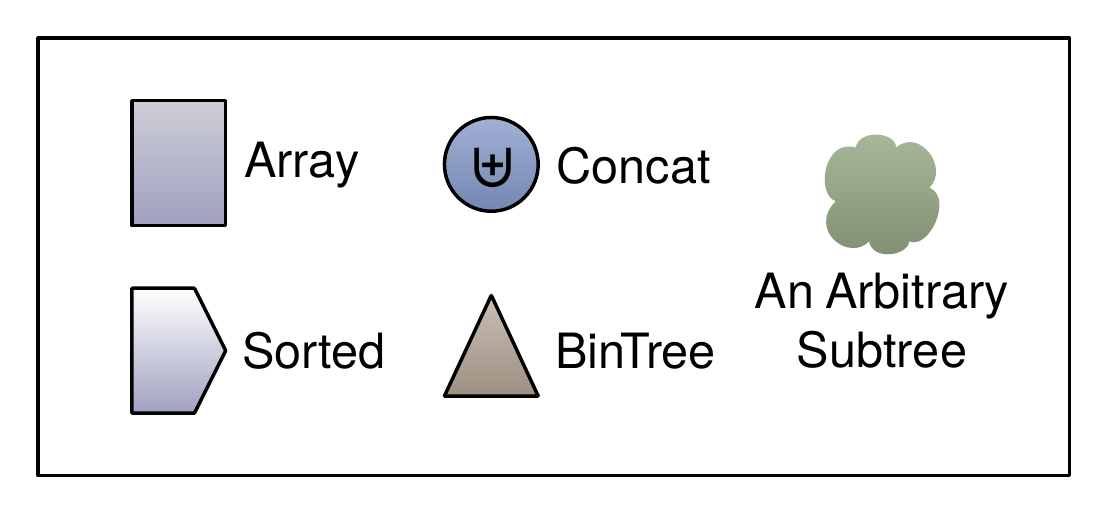}
\vspace*{-5mm}
\caption{Node types in a \IndexName}
\label{fig:nodekey}
\trimfigurespacing
\end{figure}

A Self-Adjusting Index (\IndexName) is a key-value style primary (clustered) index storing a collection of records, each (non-uniquely) identified by a key with a well defined sort order.
As illustrated in Figure~\ref{fig:system}, a \IndexName consists of three parts: an index, an optimizer, and a policy simulator.
The \IndexName's index is a tree rooted at a node designated \RootName.
Following just-in-time data structures, \IndexNames use four types of nodes, summarized in Figure~\ref{fig:nodekey}: 
(1)~\ArrayAtom: A leaf node storing an unsorted array of records,
(2)~\SortedAtom: A leaf node storing a \emph{sorted} array of records,
(3)~\ConcatAtom: An inner node pointing to two additional nodes, and
(4)~\BTreeAtom: A binary tree node that segments records in the two nodes it points to by a separator value.

The second component of \IndexName is a just-in-time optimizer, an asynchronous process that incrementally reorganizes the index, progressively rewriting its component parts to adapt it to the currently running workload.
These rewrites are guided by a \emph{policy}, a set of rules for identifying points in the index to be rewritten and for determining what rewrites to apply.
To help users to select an appropriate policy, \IndexName includes a policy simulator that generates predicted performance over time curves for specific policies.
This simulator can be used to quickly compare policies, helping users to select the policy that best meets the user's requirements for latency, preparation time, or throughput.

\begin{figure}
\centering
\includegraphics[width=0.7\columnwidth]{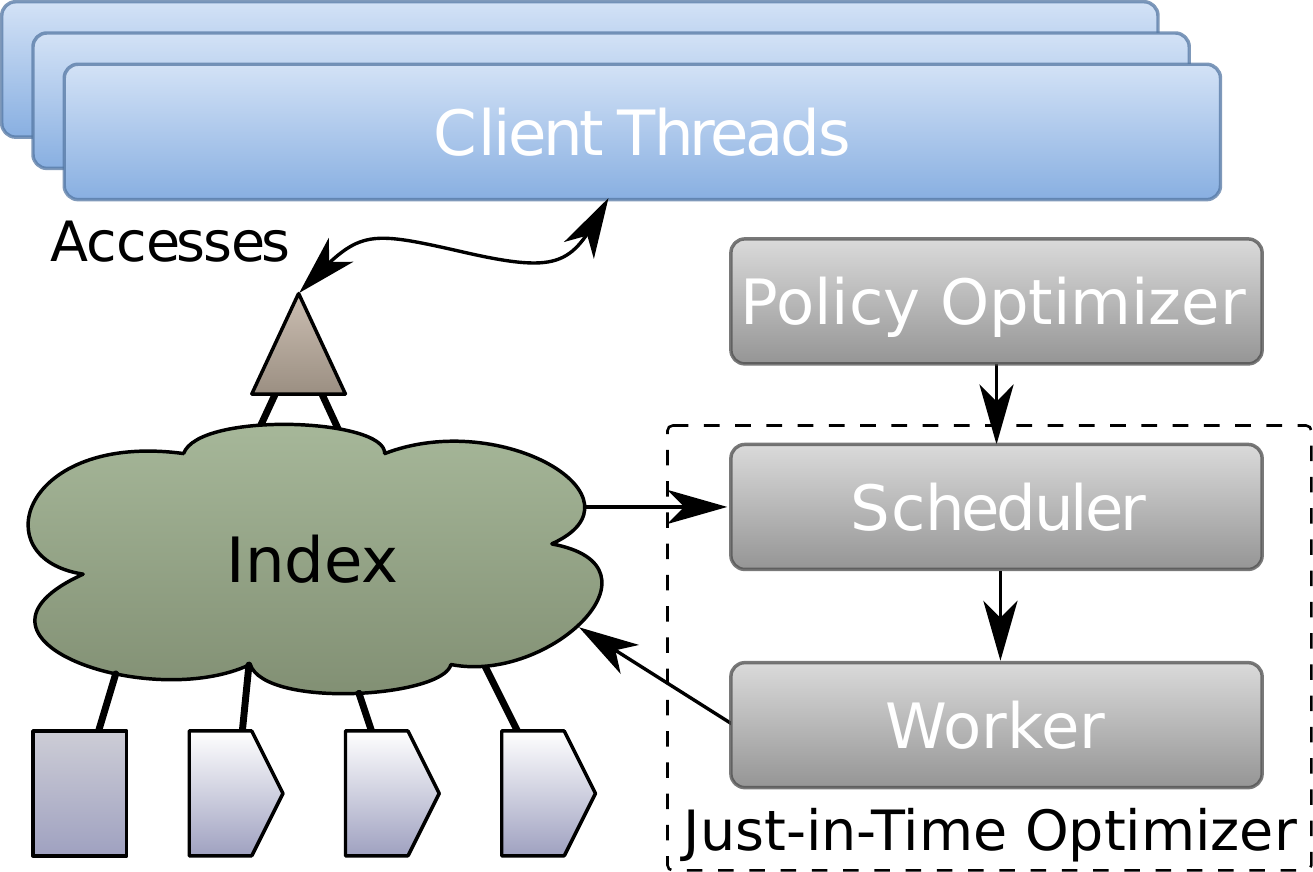}
\caption{A \IndexName}
\label{fig:system}
\trimfigurespacing
\end{figure}

\subsection{Access Paths}
A \IndexName provides lock-free access to its contents through \emph{access paths} that recursively traverse the index: 
(1) \texttt{get(key)} returns the first record with a target key,
(2) \texttt{iterator(lower)} returns an un-ordered iterator over records with keys greater than or equal to \texttt{lower}, and 
(3) \texttt{ordered\_iterator(lower)} returns an iterator over the same records, but in key order.
As an example, Algorithm~\ref{alg:get} implements the first of these access paths by recursively descending through the index.
Semantic constraints on the layout provided by \SortedAtom{} and \BTreeAtom{} are exploited where they are available.

\begin{algorithm}
\begin{algorithmic}
\caption{\texttt{Get($C$, $\KeyInstance$)}}
\label{alg:get}
\REQUIRE \texttt{$C$}: A \IndexName node
\hspace{10mm} \texttt{$\KeyInstance$}: A key
\ENSURE \texttt{$\RecordInstance$}: A record with key $\KeyInstance$ or $\textbf{None}$ if none exist.
\IF{$C \textbf{ matches } \ArrayAtom(\vec \RecordInstance)$}
  \STATE \textbf{return} $\texttt{linearScan}(\KeyInstance, \vec \RecordInstance)$
\ELSIF{$C \textbf{ matches } \SortedAtom(\vec \RecordInstance)$}
  \STATE \textbf{return} $\texttt{binarySearch}(\KeyInstance, \vec \RecordInstance)$
\ELSIF{$C \textbf{ matches } \ConcatAtom(C_1, C_2)$}
  \STATE $\RecordInstance = \texttt{Get($C_1$, \KeyInstance)}$
  \STATE \textbf{if} $\RecordInstance \neq \textbf{None}$ \textbf{then return} $\RecordInstance$
  \STATE \textbf{else return} $\texttt{Get($C_2$, \KeyInstance)}$
\ELSIF{$C \textbf{ matches } \BTreeAtom(\KeyInstance', C_1, C_2)$}
  \STATE \textbf{if} $\KeyInstance' \KeyOrder \KeyInstance$ \textbf{then return} $\texttt{Get($C_2$, \KeyInstance)}$
  \STATE \textbf{else return} $\texttt{Get($C_1$, \KeyInstance)}$
\ENDIF
\end{algorithmic}
\end{algorithm}
\trimfigurespacing

\subsection{Updates}
Organizational effort in a \IndexName is entirely offloaded to the just-in-time optimizer.
Client threads performing updates do the minimum work possible to register their changes.
To insert, the updating thread instantiates a new \ArrayAtom{} node $C$ and creates a subtree linking it and the current index root: 
\vspace*{-2mm}
$$\ConcatAtom(\RootName, C)\vspace*{-2mm}$$
This subtree becomes a new version of the root.
Although only one thread may update the index at a time, updates can proceed concurrently with the background worker thread.
This is achieved through a layer of indirection called a \emph{handle} that 
we introduce and discuss further in \Cref{sec:implementation}.

\subsection{Organization and Policy}
The background worker thread is responsible for iteratively rewriting fragments of the index into (hopefully) more efficient forms.
It needs (1) to identify fragments of the structure that need to be rewritten, (2) to decide how to rewrite those fragments, and (3) to decide how to prioritize these tasks.
We address the first two challenges by defining a fixed set of transformations for \IndexName.  
Like rewrite rules in an optimizing compiler, transformations replace subtrees of the grammar with logically equivalent structures.
Following this line of thought, we first develop a formalism that treats the state of the index at any point in time as a sentence in a grammar over the four node types.
We show that transformations can be expressed as structural rewrites over this grammar, and that for any sentence (i.e., index instance) we can enumerate the sentence fragments to which a transformation can be successfully applied.
A policy that balances the trade-offs between different types of transformations is then defined to prioritize which transforms should be applied and when.  

\subsection{Paper Outline}

\noindent The remainder of this paper is organized as follows.

\tinysection{Encoding hybrid index structures}
In Section~\ref{sec:grammar}, we introduce and formalize the \CogName grammar and show how it allows us to encode a wide range of tree-structured physical data layouts.
These include restricted sub-grammars that capture, for example, singly linked lists or binary trees.  
The grammar can express transitional physical layouts that combine elements of multiple classes of data structure.

\tinysection{Data structure transitions as an algebra}
We next outline an algebra over the \CogName grammar in Section~\ref{sec:transformations}.  
Specifically we introduce the concept of transforms, rewrites on the structure of a sentence in \CogName that preserve logical equivalence and syntactic constraints over the structure.

\tinysection{Combining transforms into a policy}
Next, in Section~\ref{sec:policies}, we show how sequences of transforms, guided by a policy, may be used to incrementally re-organize an index.  
In order to remain competitive with classical index structures, policies need to make split-second decisions on which transforms to apply.
Accordingly, we identify a specific class of local hierarchical policies that can be implemented via an incrementally maintained priority queue that tracks organizational goals and efficiently selects transforms.

\tinysection{Implementation and runtime}
After providing a theoretical basis for \IndexNames, we describe how we addressed key challenges in implementing them, outline the primary components of the \IndexName runtime, and provide an illustrative example policy: Crack-or-Sort.

\tinysection{Policy optimization}
Section~\ref{sec:model} introduces a \IndexName simulator. This simulator emulates the evolution of a \IndexName, allowing us to efficiently determine which of a range of alternative policies best meets user-provided performance goals for transitioning between index structures.

\tinysection{Assessing \IndexName's generality}
Section~\ref{sec:discussion} uses a taxonomy of index data structures proposed in \cite{DBLP:conf/sigmod/IdreosZHKG18} to evaluate \IndexName's generality.  
We propose three ideas for future work that could fully generalize 19 of the 22 design dimensions identified.

\tinysection{Evaluation}
Finally, in Section~\ref{sec:evaluation}, we assess the performance overheads \IndexNames, relative to both commonly used and state-of-the art in-memory indexes.

% We begin with a simple overview of the structure of a \systemname index, as illustrated in Figure~\ref{fig:system}.
% The central part of the index is an instance: The physical layout corresponding to a sentence in \COG.  
% One or more client threads interact with the instance by querying it or by modifying the root of the instance to create a new version of the index structure.

% Client threads perform no organizational tasks on the index at all.  

% Simultaneously, a scheduler priori reorganization tasks that are delegated to worker threads.  
% Worker threads assemble new versions of a 

%  worker thread identifies points in the structure 

% The instance is queried by one or more client threads, 
% This instance is queries 

%% file: sections/grammar.tex
% -*- root: ../main.tex -*-

Each record $\RecordInstance \in \RecordType$ is accessed exclusively by a (potentially non-unique) identifier $\KeyOf{\RecordInstance} \in \KeyType$.
We assume a total order $\KeyOrder$ is defined over elements of $\KeyType$.
We abuse syntax and use records and keys interchangeably with respect to the order, writing $\RecordInstance \KeyOrder \KeyInstance$ to mean $\KeyOf{\RecordInstance} \KeyOrder \KeyInstance$.
We write $\ArrayOfType{\GenericType}$, $\SetOfType{\GenericType}$, and $\BagOfType{\GenericType}$ to denote the type of arrays, sets, and bags (respectively) with elements of type $\GenericType$.
We write $\ArrayOf{\RecordInstance_1, \ldots, \RecordInstance_N}$ (resp., $\SetOf{\ldots}$, $\BagOf{\ldots}$) to denote an array (or set or bag) with elements $\RecordInstance_1, \ldots, \RecordInstance_N$.

To support incremental index transitions, we need a way to represent intermediate states of an index, part way between one physical layout and another.
In this section we propose a compositional organizational grammar (\CogName) that will allow us to reason about the state of a \IndexName, and the correctness of its state transitions.

\subsection{Notation and Definitions}

The atoms of \CogName are defined by four symbols $\ArrayAtom$, $\SortedAtom$, $\ConcatAtom$, $\BTreeAtom$.
A \CogName \emph{instance} is a sentence in \CogName, defined by the grammar $\CogType$ as follows:
\begin{eqnarray*}
\CogType &=&       \ArrayAtom(\ArrayOfType{\RecordType})
         \;\;|\;\; \SortedAtom(\ArrayOfType{\RecordType})\\
         &|&       \ConcatAtom(\CogType, \CogType)
         \;\;|\;\; \BTreeAtom(\KeyType, \CogType, \CogType)
\end{eqnarray*}

Atoms in \CogName map directly to physical building blocks of a data structure, while atom instances correspond to instances of a data structure or one of its sub-structures.
For example an instance of $\ArrayAtom$ represents an array of records laid out contiguously in memory, while $\ConcatAtom$ represents a tuple of pointers referencing other instances.
We write $\TypeOfCogInstance{C}$ to denote the atom symbol at the root of an instance $C \in \CogType$.

\begin{example}[Linked List]
\label{ex:linkedlist}
A linked list may be defined as a syntactic restriction over \CogName as follows
\begin{eqnarray*}
\mathcal{LL} &=& \ConcatAtom(\ArrayAtom(\ArrayOfType{\RecordType}), \mathcal{LL}) \;\;|\;\; \ArrayAtom(\ArrayOfType{\RecordType})
\end{eqnarray*}
A linked list is either a concatenation of an array (with one element by convention), and a pointer to the next element, or a terminal array (with no elements by convention).
\end{example}

Two different instances, corresponding to different representations may still encode the same data.
We describe the logical contents of an instance $C$ as a bag, denoted by $\ContentsOfCog{C}$, and use this term to define logical equivalence between two instances.
{\small
\begin{eqnarray*}
\ContentsOfCog{C} = \begin{cases}
\BagOf{\RecordInstance_1, \ldots, \RecordInstance_N} & \textbf{if } C = \ArrayAtom(\ArrayOf{\RecordInstance_1, \ldots, \RecordInstance_N})\\
\BagOf{\RecordInstance_1, \ldots, \RecordInstance_N} & \textbf{if } C = \SortedAtom(\ArrayOf{\RecordInstance_1, \ldots, \RecordInstance_N})\\
\ContentsOfCog{C_1} \uplus \ContentsOfCog{C_2} & \textbf{if } C = \ConcatAtom(C_1, C_2)\\
\ContentsOfCog{C_1} \uplus \ContentsOfCog{C_2} & \textbf{if } C = \BTreeAtom(\Wildcard, C_1, C_2)\\
\end{cases}
\end{eqnarray*}
}

\begin{definition}[Logical Equivalence]
\label{def:logicalEquivalence}
Two instances $C_1$ and $C_2$ are \emph{logically equivalent} if and only if $\ContentsOfCog{C_1} = \ContentsOfCog{C_2}$.  To denote logical equivalence we write $C_1 \IsLogicallyEquivalentTo C_2$.
\end{definition}

We write $\DescendantsOfCog{C}$ to denote the bag consisting of the instance $C$ and its descendants.
{\small
$$
\DescendantsOfCog{C} = 
\begin{cases}
\DescendantsOfCog{C_1} \uplus \DescendantsOfCog{C_2} \uplus \BagOf{C}  & \textbf{if } C = \ConcatAtom(C_1, C_2)\\
\DescendantsOfCog{C_1} \uplus \DescendantsOfCog{C_2} \uplus \BagOf{C}  & \textbf{if } C = \BTreeAtom(\Wildcard, C_1, C_2)\\
\BagOf{C} & \textbf{otherwise}
\end{cases}
$$
}

\begin{proposition}
\label{thm:descendantsOfCAreFinite}
The set $C^*$ is finite for any $C$.
\end{proposition}
% \begin{proof}
% Any sentence in \CogName is a tree
% \end{proof}

\subsection{\CogName Semantics}

$\ArrayAtom$ and $\ConcatAtom$ represent the physical layout of elements of a data structure.
The remaining two atoms provide provide semantic constraints (using the identifier order $\KeyOrder$) over the physical layout that can be exploited to make the structure more efficient to query.
We say that instances satisfying these constraints are \emph{structurally correct}.

\begin{definition}[Structural Correctness]
\label{def:structuralcorrectness}
We define the structural correctness of an instance $C \in \CogType$ (denoted by the unary relation $\IsStructurallyCorrect{C}$) for each atom individually:
\begin{enumerate}[label=\textbf{Case \arabic*}.,itemindent=*,leftmargin=13mm,labelwidth=12mm]
  \item $\ArrayAtom$ instance is structurally correct.
  \item The instance $\ConcatAtom(C_1, C_2)$ is structurally correct if and only if $C_1$ and $C_2$ are both structurally correct.
  \item \label{def:sortedIsStructurallyCorrect} The instance $\SortedAtom(\ArrayOf{r_1, \ldots, r_N})$ is structurally correct if and only if $\forall 0 \leq i < j \leq N \;:\; \RecordInstance_i \KeyOrder \RecordInstance_j$
  \item \label{def:btreeIsStructurallyCorrect} The instance $\BTreeAtom(\KeyInstance, C_1, C_2)$ is structurally correct if and only if both $C_1$ and $C_2$ are structurally correct, and if $\forall \RecordInstance_1 \in \ContentsOfCog{C_1} : \RecordInstance_1 \KeyOrderStrict \KeyInstance$ and $\RecordInstance_2 \in \ContentsOfCog{C_2}\;:\; \KeyInstance \KeyOrder \RecordInstance_2 $.
\end{enumerate}
\end{definition}

In short, $\SortedAtom$ is structurally correct if it represents a sorted array.
Similarly, $\BTreeAtom$ is structurally correct if it corresponds to a binary tree node, with its children partitioned by its identifier.  Both $\ConcatAtom$ and $\BTreeAtom$ additionally require that their children be structurally correct.

\begin{example}[Binary Tree]
A binary tree may be defined as a syntactic restriction over \CogName as follows
\begin{eqnarray*}
\mathcal{B} &=& \BTreeAtom(\KeyType, \mathcal{B}, \mathcal{B}) \;\;|\;\; \ArrayAtom(\ArrayOfType{\RecordType})
\end{eqnarray*}
A binary tree is a hierarchy of $\BTreeAtom$ inner nodes, over $\ArrayAtom$ leaf nodes (containing one element by convention).
\end{example}

%% file: sections/transformations.tex
% -*- root: ../main.tex -*-
%!TEX root = ../main.tex

\begin{figure*}
{\footnotesize
\begin{eqnarray*}
\hspace*{0.5\textwidth}
\UnsortTransform(C) &=& \begin{cases}
    \ArrayAtom(\;\vec r\;)
    & \textbf{if } C = \SortedAtom(\;\vec r\;)\\
    C & \textbf{otherwise}
  \end{cases}\\
\end{eqnarray*}
\vspace*{-20mm}
\begin{eqnarray*}
\SortTransform(C) &=& \begin{cases}
    \SortedAtom(\SortFunction(\;\vec r\;))
    & \textbf{if } C = \ArrayAtom(\;\vec r\;)\\
    C & \textbf{otherwise}
  \end{cases}\hspace*{0.5\textwidth}\\
\DivideTransform(C) &=& \begin{cases}
    \ConcatAtom(\ArrayAtom(\ArrayOf{\RecordInstance_1 \ldots \RecordInstance_{\FloorOf{\frac{N}{2}}}}), \ArrayAtom(\ArrayOf{\RecordInstance_{\FloorOf{\frac{N}{2}}+1} \ldots \RecordInstance_N}))
    & \textbf{if } C = \ArrayAtom(\ArrayOf{\RecordInstance_1 \ldots \RecordInstance_N})\\
    C & \textbf{otherwise}
  \end{cases}\\
\CrackTransform(C) &=& \begin{cases}
    \BTreeAtom(\KeyOf{\RecordInstance_{\FloorOf{\frac{N}{2}}}},
      \ArrayAtom(\ArrayOf{\RecordInstance_i \;\big |\; \RecordInstance_i \KeyOrderStrict \RecordInstance_{\FloorOf{\frac{N}{2}}}}), 
      \ArrayAtom(\ArrayOf{\RecordInstance_i \;\big |\; \RecordInstance_{\FloorOf{\frac{N}{2}}} \KeyOrder \RecordInstance_i})
    & \textbf{if } C = \ArrayAtom(\ArrayOf{\RecordInstance_1 \ldots \RecordInstance_N}))\\
    C & \textbf{otherwise}    
  \end{cases}\\
\MergeTransform(C) &=& \begin{cases}
    \ArrayAtom(\ArrayOf{\RecordInstance_1 \ldots \RecordInstance_N, \RecordInstance_{N+1} \ldots \RecordInstance_M})
    & \textbf{if } C = \ConcatAtom(\ArrayAtom(\ArrayOf{\RecordInstance_1 \ldots \RecordInstance_N}), \ArrayAtom(\ArrayOf{\RecordInstance_{N+1} \ldots \RecordInstance_M}))\\
    \ArrayAtom(\ArrayOf{\RecordInstance_1 \ldots \RecordInstance_N, \RecordInstance_{N+1} \ldots \RecordInstance_M})
    & \textbf{if } C = \BTreeAtom(\Wildcard, \ArrayAtom(\ArrayOf{\RecordInstance_1 \ldots \RecordInstance_N}), \ArrayAtom(\ArrayOf{\RecordInstance_{N+1} \ldots \RecordInstance_M}))\\
    C & \textbf{otherwise}
  \end{cases}\\
\PivotLeftTransform(C) &=& \begin{cases}
    \ConcatAtom(\ConcatAtom(C_1, C_2), C_3)
      & \textbf{if } C = \ConcatAtom(C_1, \ConcatAtom(C_2, C_3))\\
    \BTreeAtom(\KeyInstance_2, \BTreeAtom(\KeyInstance_1, C_1, C_2), C_3)
      & \textbf{if } C = \KeyInstance_1 \KeyOrderStrict \KeyInstance_2 \textbf{ and } \BTreeAtom(\KeyInstance_1, C_1, \BTreeAtom(\KeyInstance_2, C_2, C_3))\\
    C & \textbf{otherwise}
  \end{cases}
\end{eqnarray*}
}
\caption{
  Examples of \emph{correct} transforms.  $\SortTransform$ and $\UnsortTransform$ convert between $\ArrayAtom$ and $\SortedAtom$ and visa versa. $\CrackTransform$ and $\DivideTransform$ both fragment $\ArrayAtom$s, and both are reverted by $\MergeTransform$.  $\CrackTransform$ in particular uses an arbitrary array element to partition its input value (the $\frac{N}{2}$th element in this example), analogous to the RadixCrack operation of \cite{DBLP:journals/pvldb/IdreosMKG11}.  $\PivotLeftTransform$ rotates tree structures counterclockwise and a symmetric $\PivotRightTransform$ may also be defined.
  The function $\SortFunction : \ArrayOf{\RecordType} \rightarrow \ArrayOf{\RecordType}$ returns a transposition of its input sorted according to $\KeyOrder$.
}
\label{fig:transformExamples}
\end{figure*}

We next formalize state transitions in a \IndexName through pattern-matching rewrite rules over \CogName called transforms.

\begin{definition}[Transform] 
We define a transform $\TransformInstance$ as any member of the family $\TransformType$ of endomorphisms over \CogName instances.  Equivalently, any transform $\TransformInstance \in \TransformType$ is a morphism $\TransformInstance : \CogType \rightarrow \CogType$ from instance to instance.
\end{definition}

Figure~\ref{fig:transformExamples} illustrates a range of common transforms that correspond to common operations on index structures.
For consistency, we define transforms over all instances and not just instances where the operation ``makes sense.''
On other instances, transforms behave as the identity ($\IdentityTransform(C) = C$).

Clearly not all possible transforms are useful for organizing data.
For example, the well defined, but rather unhelpful transform $\textbf{Empty}(C) = \ArrayAtom(\ArrayOf{})$ transforms any \CogName instance into an empty array.
To capture this notion of a ``useful'' transform, we define two correctness properties: structure preservation and equivalence preservation.

\begin{definition}[Equivalence Preserving Transforms]
A transform $\TransformInstance$ is defined to be \emph{equivalence preserving} if and only if  $\forall C : C \IsLogicallyEquivalentTo \TransformInstance(C)$ (Definition~\ref{def:logicalEquivalence}).
\end{definition}
\begin{definition}[Structure Preserving Transforms]
A transform $\TransformInstance$ is defined to be \emph{structure preserving} if and only if $\forall C : \IsStructurallyCorrect{C}\implies \IsStructurallyCorrect{\TransformInstance(C)}$ (Definition~\ref{def:structuralcorrectness}).
\end{definition}

A transform is equivalence preserving if it preserves the logical content of the instance.  It is structure preserving if it preserves the structure's semantic constraints (e.g., the record ordering constraint on instances of the $\SortedAtom$ atom).  If it is both, we say that the transform is correct.

\begin{definition}[Correct Transform]
\label{def:correctTransform}
We define a transform $\TransformInstance$ to be \emph{correct} (denoted $\TransformIsCorrect{\TransformInstance}$) if $\TransformInstance$ is both structure and equivalence preserving.
\end{definition}

In \Cref{apx:ExampleTransformCorrectness} we give proofs of correctness for each of the transforms in \Cref{fig:transformExamples}.

%%%%%%%%%%%%%%%%%%%%%%%%%%%%%%%%%%%%%%%%%%%%%%%%%%%%%%%%%%%%%%%%%%%%
%%%%%%%%%%%%%%%%%%%%%%%%%%%%%%%%%%%%%%%%%%%%%%%%%%%%%%%%%%%%%%%%%%%%
%%%%%%%%%%%%%%%%%%%%%%%%%%%%%%%%%%%%%%%%%%%%%%%%%%%%%%%%%%%%%%%%%%%%
%%%%%%%%%%%%%%%%%%%%%%%%%%%%%%%%%%%%%%%%%%%%%%%%%%%%%%%%%%%%%%%%%%%%

%%%%%%%%%%%%%%%%%%%%%%%%%%%%%%%%%%%%%%%%%%%%%%%%%%%%%%%%%%%%%%%%%%%%
%%%%%%%%%%%%%%%%%%%%%%%%%%%%%%%%%%%%%%%%%%%%%%%%%%%%%%%%%%%%%%%%%%%%
%%%%%%%%%%%%%%%%%%%%%%%%%%%%%%%%%%%%%%%%%%%%%%%%%%%%%%%%%%%%%%%%%%%%
%%%%%%%%%%%%%%%%%%%%%%%%%%%%%%%%%%%%%%%%%%%%%%%%%%%%%%%%%%%%%%%%%%%%

\subsection{Meta Transforms}

Transforms such as those illustrated in Figure~\ref{fig:transformExamples} form the atomic building blocks of a policy for re-organizing data structures.  
For the purposes of this paper, we refer to these six transforms, together with $\PivotRightTransform$ and the identity transform $\IdentityTransform$, collectively as the \emph{atomic transforms}, denoted $\AtomicTransformSet$.
We next introduce a framework for constructing more complex transforms from these building blocks.

\begin{definition}[Composition]
For any two transforms $\TransformInstance_1, \TransformInstance_2 \in \TransformType$, we denote by $\TransformInstance_1 \circ \TransformInstance_2$ the composition of $\TransformInstance_1$ and $\TransformInstance_2$:
$$(\TransformInstance_1 \circ \TransformInstance_2)(C) \IsDefinedAs \TransformInstance_2(\TransformInstance_1(C))$$
\end{definition}

Transform composition allows us to build more complex transforms from the set of atomic transforms.  We also consider \emph{meta transforms} that manipulate transform behavior. 

\begin{definition}[Meta Transform]
A meta transform $\MetaTransformInstance$ is any correctness-preserving endofunctor over the set of transforms.
That is, any functor $\MetaTransformInstance : \TransformType \rightarrow \TransformType$ is a meta transform if and only if $\forall \TransformInstance \in \TransformType : \TransformIsCorrect{\TransformInstance} \implies \TransformIsCorrect{\MetaTransformInstance[\TransformInstance]}$ (Definition~\ref{def:correctTransform}).
\end{definition}

\noindent We are specifically interested in two meta transforms that will allow us to apply transforms not just to the root of an instance, but to any of its descendants as well.
\resizebox{\columnwidth}{!}{
$\LHSMetaTransform[\TransformInstance](C) = \begin{cases}
\ConcatAtom(\TransformInstance(C_1), C_2) 
  & \textbf{if } C = \ConcatAtom(C_1, C_2)\\
\BTreeAtom(\KeyInstance, \TransformInstance(C_1), C_2) 
  & \textbf{if } C = \BTreeAtom(\KeyInstance, C_1, C_2)\\
C & \textbf{otherwise}
\end{cases}$
}
\resizebox{\columnwidth}{!}{
$\RHSMetaTransform[\TransformInstance](C) = \begin{cases}
\ConcatAtom(C_1, \TransformInstance(C_2)) 
  & \textbf{if } C = \ConcatAtom(C_1, C_2)\\
\BTreeAtom(\KeyInstance, C_1, \TransformInstance(C_2)) 
  & \textbf{if } C = \BTreeAtom(\KeyInstance, C_1, C_2)\\
C & \textbf{otherwise}
\end{cases}$
}
\begin{theorem}[$\LHSMetaTransform$ and $\RHSMetaTransform$ are meta transforms]
\label{thm:lhsrhsAreMeta}
$\LHSMetaTransform$ and $\RHSMetaTransform$ are correctness-preserving endofunctors over $\TransformType$.
\end{theorem}
The proof, given in \Cref{apx:lhsrhsAreMeta}, is a simple structural recursion over cases.

We refer to the closure of $\LHSMetaTransform$ and $\RHSMetaTransform$ over the atomic transforms as the set of \emph{hierarchical transforms}, denoted $\HierarchicalTransforms$.
$$\HierarchicalTransforms = \AtomicTransformSet 
\cup \Comprehension{\LHSMetaTransform[\TransformInstance]}{\TransformInstance \in \HierarchicalTransforms} 
\cup \Comprehension{\RHSMetaTransform[\TransformInstance]}{\TransformInstance \in \HierarchicalTransforms}$$

\begin{corrolary}
Any hierarchical transform is correct.
\end{corrolary}

%% file: sections/policies.tex
% -*- root: ../main.tex -*-

Transforms give us a means of manipulating instances, but to actually allow an index to transition from one form to another we need a set of rules, called a \emph{policy}, to dictate which transform to apply and when.
We begin by defining policies broadly, before refining them into an efficiently implementable family of \emph{enumerable score-based policies}.

\begin{definition}[Policy]
A policy $\PolicyInstance$ is defined by the 2-tuple $\PolicyInstance = \Tuple{\PolicyDomain, \PolicyHeuristic}$, where the policy's \emph{domain} $\PolicyDomain \subseteq \TransformType$ is a set of transforms and $\PolicyHeuristic : \ \CogType \rightarrow \PolicyDomain$ is a heuristic function that selects one of these transforms to apply to a given instance.
\end{definition}

A policy guides the transition of an index from an instance representing its initial state to a final state achieved by repeatedly applying transforms selected by the heuristic $\PolicyHeuristic$. 
We call the sequence of instances reached in this way a \emph{trace}.

\begin{definition}[Trace]
The trace of a policy $\PolicyInstance = \Tuple{\PolicyDomain, \PolicyHeuristic}$ on instance $C_0$, denoted $\TraceOf(\PolicyInstance, C_0)$, is defined as the infinite sequence of instances $\ArrayOf{C_0, C_1, \ldots}$ starting with $C_0$, and with subsequent instances $C_i$ obtained as:%by applying the transform selected by $\PolicyHeuristic$ for instance $C_{i-1}$. Specifically, for $i > 0$
$$C_i \IsDefinedAs \TransformInstance_i(C_{i-1}) \;\;\;\textbf{  where  }\;\;\; \TransformInstance_i = \PolicyHeuristic(C_{i-1})$$
\end{definition}

Although traces are infinite, we are specifically interested in policies with traces that reach a steady (fixed point) state.
We say that such a trace (resp., any policy guaranteed to produce such a trace) is \emph{terminating}.

\begin{definition}[Terminating Trace, Policy]
A trace $\ArrayOf{C_1, C_2, \ldots}$ \emph{terminates} after $N$ steps if and only if $\forall i,j > N : C_i = C_j,$.
A policy $\PolicyInstance$ is \emph{terminating} when
$$\forall C \exists N : \TraceOf(\PolicyInstance, C) \textbf{ terminates after $N$ steps}$$
\end{definition}

A policy's domain may be large, or even infinite as in the case of the hierarchical transforms.
However, only a much smaller fragment will typically be useful for any specific instance.
We call this fragment the \emph{active domain}.

\begin{definition}
The \emph{active domain} of a policy $\Tuple{\PolicyDomain, \PolicyHeuristic}$, relative to an instance $C$ (denoted $\PolicyDomain_C$) is the subset of the policy's domain that does not behave as the identity on $C$.
$$\PolicyDomain_C \IsDefinedAs \Comprehension{\TransformInstance}{\TransformInstance \in \PolicyDomain \wedge \TransformInstance(C) \neq C}$$
\end{definition}

%%%%%%%%
\HiddenProof{
  \begin{definition}[Acyclic Trace]
  A trace is acyclic iff it terminates after N steps and $\forall i,j \leq N$ and $i \neq j$ 
  $$C_i \neq C_j $$
  \end{definition}
  \begin{definition}[Finite Trace]
  A $Trace((D,H),C)$ is finite if $\exists N$ for which the trace is acyclic-N
  \end{definition}
  \begin{definition}[Acyclic Policy]
  A policy $(D, H)$ is acyclic iff $\forall C,$ $Trace((D,H),C)$ is finite. 
  \end{definition}
}

%%%%%%%%%%%%%%%%%%%%%%%%%%%%%%%%%%%%%%%%%%%%%%%%%%%%%%%%%%%%%%%%%%%%%%%%
%%%%%%%%%%%%%%%%%%%%%%%%%%%%%%%%%%%%%%%%%%%%%%%%%%%%%%%%%%%%%%%%%%%%%%%%
%%%%%%%%%%%%%%%%%%%%%%%%%%%%%%%%%%%%%%%%%%%%%%%%%%%%%%%%%%%%%%%%%%%%%%%%
%%%%%%%%%%%%%%%%%%%%%%%%%%%%%%%%%%%%%%%%%%%%%%%%%%%%%%%%%%%%%%%%%%%%%%%%
\subsection{Bounding the Active Domain}

A policy's heuristic function will be called numerous times in the course of an index transition, making it a prime candidate for performance optimization.
We next explore one particular family of policies that admit a stateful, incremental implementation of their heuristic function.
This approach treats the heuristic function as a ranking query over the active domain, selecting the most appropriate (highest scoring) transform at any given time.  
However, rather than recomputing scores at every step, we incrementally maintain a priority queue over the active domain.
For this incremental approach to be feasible, we need to ensure that only a finite (and ideally small) number of scores change with each step.

\begin{definition}[Enumerable Policy] 
A policy $\Tuple{\PolicyDomain, \PolicyHeuristic}$ is enumerable if and only if its active domain is finite for every finite instance $C$, or equivalently when
$\forall C : |\PolicyDomain_C| \in \mathbb N$
\end{definition}

We are particularly interested in policies that use the hierarchical transforms as their domain.  We also refer to such policies as \emph{hierarchical}.  
In order to show that hierarchical policies are enumerable, we first define a utility \emph{target} function that ``unrolls'' an arbitrarily deep stack of $\LHSMetaTransform$ and $\RHSMetaTransform$ meta transforms.
The target function returns (1) The atomic transform at the base of the stack of meta transforms and (2) the descendant that this atomic transform would be applied to.

\begin{definition}
\label{def:targetFunction}
Given a hierarchical policy $\Tuple{\HierarchicalTransforms, \PolicyHeuristic}$ and an instance $C$, let the \emph{target function} $\DescendantsOfCog{f}_C : \PolicyDomain_C \rightarrow (\DescendantsOfCog{C} \times \AtomicTransformSet)$ of the policy on $C$ is defined as follows
{\small
$$ 
\DescendantsOfCog{f}_C(\TransformInstance) \IsDefinedAs 
\begin{cases}
\Tuple{C, \TransformInstance} & \textbf{if } \TypeOfCogInstance{C} \in \SetOf{\ArrayAtom, \SortedAtom} \\
\Tuple{C, \TransformInstance} & \textbf{else if } \TransformInstance \in \AtomicTransformSet \\
\DescendantsOfCog{f}_{C_1}(\TransformInstance') & \textbf{else if } \TransformInstance = \LHSMetaTransform[\TransformInstance'] \\
\DescendantsOfCog{f}_{C_2}(\TransformInstance') & \textbf{else if } \TransformInstance = \RHSMetaTransform[\TransformInstance'] \\
\end{cases}
$$
}
\end{definition}

\begin{lemma}[Injectivity of $\DescendantsOfCog{f}_C$]
\label{thm:targetFnIsInjective}
The target function $\DescendantsOfCog{f}_C$ of any hierarchical policy $\Tuple{\HierarchicalTransforms, \PolicyHeuristic}$ for any instance $C$ is injective.
\end{lemma}
\begin{proof}
By recursion over $C$.
The base case occurs when $\TypeOfCogInstance{C} \in \SetOf{\ArrayAtom, \SortedAtom}$.  
In this case $\DescendantsOfCog{C} = \BagOf{C}$.
Furthermore, 
$\forall \TransformInstance : \LHSMetaTransform[T] = \RHSMetaTransform[T] = \IdentityTransform$
and so $\PolicyDomain_C \subseteq \AtomicTransformSet$.
The target function always follows its first case and is trivially injective.
The first recursive case occurs for $C = \ConcatAtom(C_1, C_2)$. 
By definition, a hierarchical transform can be (1) An atomic transform, (2) $\LHSMetaTransform[\TransformInstance]$, or (3) $\RHSMetaTransform[\TransformInstance]$.
Each of the latter three cases covers one of each of the three parts of the definition of a hierarchical transform.
Assuming that $\DescendantsOfCog{f}_{C_1}$ and $\DescendantsOfCog{f}_{C_2}$ are injective, $\DescendantsOfCog{f}_C$ will also be injective because each case maps to a disjoint partition of 
$\DescendantsOfCog{C} = \DescendantsOfCog{C_1} \uplus \DescendantsOfCog{C_2} \uplus \BagOf{C}$.
The proof for the second recursive case, where $\TypeOfCogInstance{C} = \BTreeAtom$ is identical.
Thus $\forall C : \DescendantsOfCog{f}_C$ is injective
\end{proof}

Using injectivity of the target function, we can show that any hierarchical policy is enumerable.

\begin{theorem}[Hierarchical policies are enumerable]
\label{thm:hierarchicalAreEnumerable}
Any hierarchical policy $\Tuple{\HierarchicalTransforms, \PolicyHeuristic}$ is enumerable.
\end{theorem}
\begin{proof}
Recall the definition of hierarchical transforms
$$\HierarchicalTransforms = \AtomicTransformSet 
\cup \Comprehension{\LHSMetaTransform[\TransformInstance]}{\TransformInstance \in \HierarchicalTransforms} 
\cup \Comprehension{\RHSMetaTransform[\TransformInstance]}{\TransformInstance \in \HierarchicalTransforms}$$
\noindent
By \Cref{thm:targetFnIsInjective}, $|\PolicyDomain_C| \leq |\DescendantsOfCog{C} \times \AtomicTransformSet| \leq |\DescendantsOfCog{C}|\times|\AtomicTransformSet|$.  
By \Cref{thm:descendantsOfCAreFinite}, $\DescendantsOfCog{C}$ is finite and the set of atomic transforms $\AtomicTransformSet$ is finite by definition
Thus, $\PolicyDomain_C$ must also be finite.
\end{proof}

Intuitively, there is a finite number of atomic transforms ($|\AtomicTransformSet$), that can be applied at a finite set of positions within $C$ ($|\DescendantsOfCog{C}$). 
Any other hierarchical transform must be idempotent, so we can (very loosely) bound the active domain of a hierarchical policy on instance $C$ by $|\DescendantsOfCog{C}|\times|\AtomicTransformSet|$

%%%%%%%%%%%%%%%%%%%%%%%%%%%%%%%%%%%%%%%%%%%%%%%%%%%%%%%%%%%%%%%%%%%%%%%%
%%%%%%%%%%%%%%%%%%%%%%%%%%%%%%%%%%%%%%%%%%%%%%%%%%%%%%%%%%%%%%%%%%%%%%%%
%%%%%%%%%%%%%%%%%%%%%%%%%%%%%%%%%%%%%%%%%%%%%%%%%%%%%%%%%%%%%%%%%%%%%%%%
%%%%%%%%%%%%%%%%%%%%%%%%%%%%%%%%%%%%%%%%%%%%%%%%%%%%%%%%%%%%%%%%%%%%%%%%
\subsection{Scoring Heuristics}

As previously noted, we are particularly interested in policies that work by scoring the set of available transforms with respect to their utility.  

\begin{definition}[Scoring Policy]
Let $\ScoreFunction : (\PolicyDomain \times \CogType) \rightarrow \mathbb{N}_0$ be a \emph{scoring function} for every transform, instance pair ($\TransformInstance, C$) that satisfies the constraint:
$(\TransformInstance(C) = C) \Rightarrow (\ScoreFunction(\TransformInstance, C) = 0)$
A \emph{scoring policy} $\Tuple{\PolicyDomain, \PolicyHeuristic_{\ScoreFunction}}$ is a policy with a heuristic function defined as
$\PolicyHeuristic_{\ScoreFunction}(C) \IsDefinedAs \textbf{argmax}_{\TransformInstance \in \PolicyDomain}(\ScoreFunction(\TransformInstance,C))$
\end{definition}

In short, a scoring heuristic policy one that selects the next transform to apply based on a scoring function $\ScoreFunction$, breaking ties arbitrarily.
Additionally, we require that transforms not in the active domain (i.e., that leave their inputs unchanged) must be assigned the lowest score (0).

As we have already established, the number of scores that need to be computed is finite and enumerable.  
However, it is also linear in the number of atoms in the instance.
Ideally, we would like to avoid recomputing all of the scores at each iteration by precomputing the scores once and then incrementally maintaining them as the instance is updated.
For this to be feasible, we also need to bound the number of scores that change with each step of the policy.
We do this by first defining two properties of policies: independence, which requires that the score of a (hierarchical) transform be exclusively dependent on its target atom (\Cref{def:targetFunction}); and locality, which further requires that the score of a transform be independent of the node's descendants past a bounded depth.  
We then show that with any scoring function that satisfies these properties, only a finite number of scores change with any transform, and consequently that the output of the scoring function on every element of the active domain can be efficiently incrementally maintained.

\begin{definition}[Independent Policy]
\label{def:independentPolicy}
Let $C^{<}$ be the set of instances with $C$ as a left child.
\begin{multline*}
C^{<} = \Comprehension{\ConcatAtom(C, C')}{C' \in \CogType} \\
\cup \Comprehension{\BTreeAtom(\KeyInstance, C, C')}{\KeyInstance \in \KeyType \wedge C' \in \CogType}
\end{multline*}
and define $C^{>}$ symmetrically as the set of instances with $C$ as a right child.
We say that a hierarchical scoring policy $\Tuple{\HierarchicalTransforms, \PolicyHeuristic_{\ScoreFunction}}$ is \emph{independent} if and only if for any $\TransformInstance$, $C$
\begin{eqnarray*}
\forall C' \in C^{<} : \ScoreFunction(\TransformInstance, C) 
&=& \ScoreFunction(\LHSMetaTransform[\TransformInstance], C')\\
\forall C' \in C^{>} : \ScoreFunction(\TransformInstance, C) 
&=& \ScoreFunction(\RHSMetaTransform[\TransformInstance], C')
\end{eqnarray*}
\end{definition}

% In general, an instance $C$ can represent not only the complete physical layout of an index, but also any sub-structure within that index.  
% An independent policy is one that only considers the instance itself, and not any other part of the overall index when assigning a score.
% Note that for a independent policy, we need only define the scores of the atomic transforms for each instance, as this implicitly defines scores of transforms derived from $\LHSMetaTransform$ and $\RHSMetaTransform$. 

\begin{definition}[Local Policy]
An independent hierarchical scoring policy $\Tuple{\HierarchicalTransforms, \PolicyHeuristic_{\ScoreFunction}}$ is local if and only if:
\begin{multline*}
\forall \TransformInstance\forall C_1 \forall C_2 \textbf{ s.t. } (\DescendantsOfCog{C_1} - \BagOf{C_1}) = (\DescendantsOfCog{C_2} - \BagOf{C_2}): \\
\ScoreFunction(\TransformInstance, C_1) = \ScoreFunction(\TransformInstance, C_2)
\end{multline*}
\end{definition}

% Local policies extend the definition of independence with the requirement that the score be based entirely on the root of the instance, and not any of the root's descendants.
% Before showing that local policies are amenable to being implemented through an incrementally maintained priority queue, 
The following definition uses the policy's target function (\Cref{def:targetFunction}) to define a weighted list of all of the policy's targets.

\begin{definition}[Weighted Targets]
Let $\Tuple{\HierarchicalTransforms, \PolicyHeuristic_{\ScoreFunction}}$ be a hierarchical scoring policy.
The weighted targets of instance $C$, denoted $\WeightedTargets_C : \BagOfType{\AtomicTransformSet \times \mathbb N_0}$ is bag of 2-tuples defined as
{\footnotesize
$$\WeightedTargets_C = 
\BagOf{\Tuple{\TransformInstance', \ScoreFunction(\TransformInstance, C')}\;\big |\;\TransformInstance \in \PolicyDomain_C \wedge (C', \TransformInstance') = \DescendantsOfCog{f}_C(\TransformInstance)}
$$
}
\end{definition}

\begin{theorem}[Bounded Target Updates] 
\label{thm:boundedTargetUpdates}
Let $\Tuple{\HierarchicalTransforms, \PolicyHeuristic_{\ScoreFunction}}$ be a local hierarchical scoring policy, $C$ be an instance, $\TransformInstance \in \PolicyDomain_C$ be a transform, and $C' = \TransformInstance(C)$.
The weighted targets of $C$ and $C'$ differ by at most $4 \times |\AtomicTransformSet|$ elements.$$\big|(\WeightedTargets_{C} \uplus \WeightedTargets_{C'}) - (\WeightedTargets_{C} \cap \WeightedTargets_{C'})\big| \leq 4 \times |\AtomicTransformSet|$$
\end{theorem}
The proof, given in \Cref{apx:boundedTargetUpdates}, is based on the observation that the independence and locality properties restrict changes to the target function's outputs to exactly the set of nodes added, removed, or modified by the applied transform, excluding ancestors or descendants.  In the worst cases (\DivideTransform, \CrackTransform, or \MergeTransform) this is 4 nodes.

\Cref{thm:boundedTargetUpdates} shows that we can incrementally maintain the weighted target list incrementally, as only a finite number of its elements change at any policy step.  
This allows us to materialize the weighted target list as a priority queue, who's first element is always the policy's next transform.

% \begin{theorem}[Active Domain Score Change Bound] 
% For a local hierarchical scoring policy $\Tuple{\HierarchicalTransforms, \PolicyHeuristic_{\ScoreFunction}}$ and an instance $C$, let $\WeightedActiveDomain_{C, \ScoreFunction}$ be the \emph{weighted active domain}, the active domain transforms annotated with their scores
% $$\WeightedActiveDomain_{C, \ScoreFunction} = \Comprehension{\Tuple{\TransformInstance, \ScoreFunction(\TransformInstance, C)}}{\TransformInstance \in \PolicyDomain_C}$$
% Let $C' = \TransformInstance(C)$ denote the result of applying any transform $\TransformInstance \in \PolicyDomain_C$.  The weighted active domains of $C$ and $C'$ differ by at most 4 elements.
% $$\big|(\WeightedActiveDomain_{C, \ScoreFunction} \cup \WeightedActiveDomain_{C', \ScoreFunction}) - (\WeightedActiveDomain_{C, \ScoreFunction} \cap \WeightedActiveDomain_{C', \ScoreFunction})\big| \leq 4$$
% \end{theorem}
% \begin{proof}
% Recall the injective function $\DescendantsOfCog{f}_C$, defined as part of \Cref{thm:hierarchicalAreEnumerable} and replicated below
% By \Cref{def:independentPolicy}, if $\DescendantsOfCog{f}_C(\TransformInstance) = \Tuple{C', \TransformInstance}$, then $\ScoreFunction(\TransformInstance, C) = \ScoreFunction(\TransformInstance, C')$

% \end{proof}

%% file: sections/implementation.tex
% -*- root: ../main.tex -*-

So far, we have introduced \CogName and shown how policies can be used to gradually reorganize a \CogName instance by repeatedly applying incremental transforms to the structure.  

In this section, we discuss the challenges in translating \IndexNames from the theory we have defined so far into practice.
As already noted, \CogName instances describe the physical layout of a \IndexName.  
We implemented each atom as a C++ class using the reference-counted \texttt{shared\_ptr}s for garbage collection.
To implement the \ArrayAtom{} and \SortedAtom{} atoms, we used the  C++ Standard Template Library \texttt{vector} class.

\subsection{Concurrency and Handles}
Because \IndexNames rely on background optimization, efficient concurrency is critical.
This motivated our choice to base the \IndexName index on functional data structures.
In a functional data structure, objects are immutable once instantiated.
Only the root may be updated to a new version, typically through an atomic pointer swap.  
Explicit versioning makes it possible for the background worker thread to construct a new version of the structure without taking out any locks in the process.
Only a short lock is required to swap in the new version.

Immutability does come with a cost: any mutations must also copy un-modified data into a new object.  
However, careful use of pointers can minimize the impact of such copies.

\begin{figure}
\centering
\includegraphics[width=0.7\columnwidth]{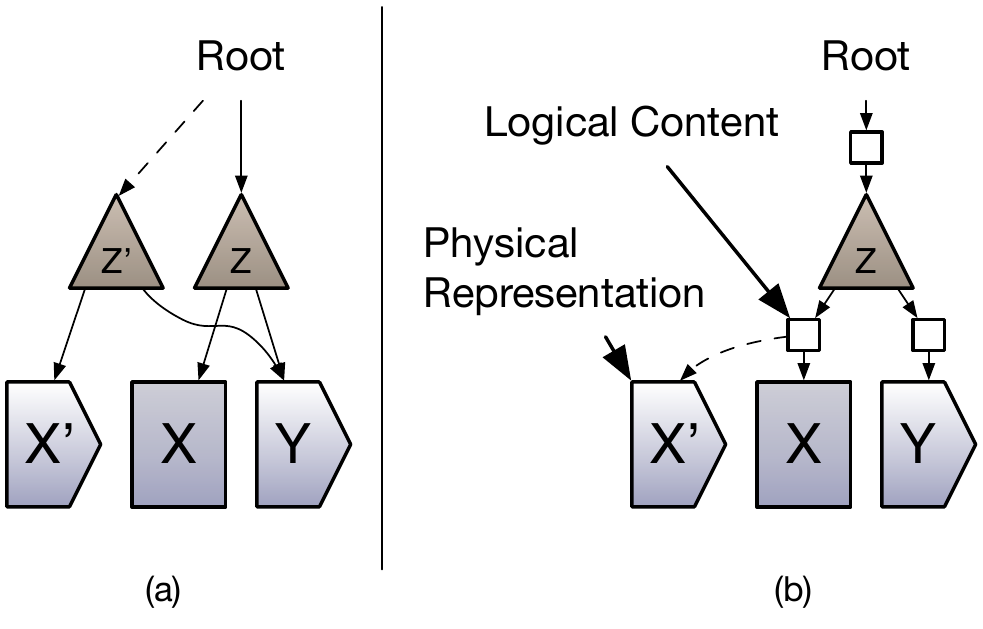}
\caption{Classical immutable data structures (a) vs with handles (b).}
\label{fig:handles}
\trimfigurespacing
\end{figure}

\begin{example}
\Cref{fig:handles}.a shows the effects of applying $\LHSMetaTransform[\SortTransform]$ to an immutable \CogName instance, replacing an unsorted $\ArrayAtom$ \texttt{X} with a $\SortedAtom$ equivalent \texttt{X'}.  
Note that in addition to replacing \texttt{X}, each of its ancestors must also be replaced.
\end{example}

Replacing a logarithmic number of ancestors is better than replacing the entire structure.
However, even a logarithmic number of new objects for every update can be a substantial expense when individual transforms can take on the order of micro-seconds. 
We avoid this overhead by using indirection to allow limited mutability under controlled circumstances.
Inspired by early forms of memory management~\cite{handles}, we define a new object called a \emph{handle}.

Handles store a pointer to a \CogName atom, and all \CogName atoms (i.e., \BTreeAtom{} and \ConcatAtom{}), as well as the root, use handles as indirect references to their children.
Handles provide clear semantics for a programmer expectations: 
A pointer to an atom guarantees \emph{physical} immutability, while a pointer to a handle guarantees only \emph{logical} immutability.  
Thus, any thread can safely replace the pointer stored in a handle with a pointer to any other logically equivalent atom.
Accordingly, we refer to such structures as \emph{semi-functional} data structures.

\begin{example}
Continuing the example, \Cref{fig:handles}.b shows the same operation performed on a structure that uses handles.  
Ancestors of the modified node are unchanged: Only the handle pointer is modified.
\end{example}

We observe that \CogName atoms can safely be implemented using handles.
The only correctness property we need to enforce is structural correctness, which depends only on the node itself and the \emph{logical} contents ($\ContentsOfCog{\cdot}$) of its descendants.
Thus only logical consistency is needed and handles suffice.

Similarly, the \LHSMetaTransform{} and \RHSMetaTransform{} and meta transform creates an exact copy of the root, modulo the affected pointer.
Furthermore the only node modified is the one reached by unrolling the stack of meta transforms, and by definition correct transforms must produce a new structure that is logically equivalent.
Thus, any hierarchical transform can be safely, efficiently applied to a \IndexName by a single modification to the handle of the target atom (\Cref{def:targetFunction}).

\newcommand{\IteratorOpPrefix}{\hspace*{-1.5mm}\textcolor{gray}{\ding{228}}\hspace{1mm}}

\subsection{Concurrent Access Paths}
We have already described the get method in \Cref{sec:overview}.
The remaining access paths instantiate iterators that traverse the tree, lazily dereferencing handles as necessary.
Un-ordered iterators provide two methods:\\

\noindent {\footnotesize
\begin{tabular}{ll}
\IteratorOpPrefix \texttt{$\RecordInstance \leftarrow$ Get()} & returns the iterator's current record\\
\IteratorOpPrefix \texttt{Step()} & advances the iterator to the next record\\
\multicolumn{2}{l}{
\hspace*{-2.3mm}
Additionally, ordered iterators provide the method:
}\\
\IteratorOpPrefix \texttt{Seek($\KeyInstance$)}& advances to the first record $\RecordInstance$ where $\RecordInstance \KeyOrderRev \KeyInstance$
\end{tabular}
}\\

\noindent For iterators over \SortedAtom{} and \ArrayAtom{} atoms, we directly use the C++ vector class's iterator.
Generating an ordered iterator over an \ArrayAtom{} atom forces a \SortTransform{} first.
Iterators for the remaining atom types lazily create a replica of the root instance using only physical references to ensure consistency.
Unordered iterators traverse trees left to right.
Ordered iterators over \ConcatAtom{} atoms are implemented using merge-sort.
We implement a special-case iterator for \BTreeAtom{}s that iterates over contiguous \BTreeAtom{}s, lazily loading nodes from their handles as needed.

\subsection{Handles and Updates}
Handles also make possible concurrency between a \IndexName's worker thread and threads updating the \IndexName.  
In keeping with the convention that structures referenced by a handle pointers can only be swapped with logically equivalent structures, a thread updating a \IndexName must replace the root handle with an entirely new handle. 
Because the worker thread will only ever swap pointers referenced by a handle, it will never undo the effects of an update.
Better still, if the old root handle is re-used as part of the new structure (as discussed in \Cref{sec:overview}), optimizations applied to the old root or any of its descendants will seamlessly be applied to the new version of the index as well.

\subsection{Transforms and the Policy Scheduler}

Our policy scheduler is optimized for local hierarchical policies.  
Policies are implemented by defining a scoring function
\vspace*{-3mm}
$$\mathbb N_0 \leftarrow \texttt{score($\TransformInstance$, $C$)} \textbf{ where } \TransformInstance \in \AtomicTransformSet\vspace*{-1mm}$$
Based on this function, the policy scheduler builds a priority queue of 3-tuples $\Tuple{\texttt{handle}, \AtomicTransformSet, \mathbb N_0}$, including a \texttt{handle} to a descendant of the root, an atomic transform to apply to the descendant instance, and the policy's score for the transform applied to the instance referenced by the handle.  
As an optimization, only the highest-scoring transform for each handle is maintained in the queue.
The scheduler iteratively selects the highest scoring transform and applies it to the structure.  Handles destroyed (resp., created) by applied transforms are removed from (resp., added to) the priority queue.
The iterator continues until no transforms remain in the queue or all remaining transforms have a score of zero, at which point we say the policy has \emph{converged}.

\subsection{Example Policy: Crack-Sort-Merge}
\label{sec:crack-or-sort}

As an example of policies being used to manage cost/benefit tradeoffs in index structures, we compare two approaches to data loading: database cracking~\cite{DBLP:conf/cidr/IdreosKM07} and upfront organization.
In a study comparing cracking to upfront indexing, Schuhknecht et. al.~\cite{DBLP:journals/pvldb/SchuhknechtJD13} observe that for workloads consisting of more than a few scans, it is faster to build an index upfront.
Here, we take a more subtle approach to the same problem.
The $\CrackTransform$ transform has lower upfront cost than the $\SortTransform$ transform (scaling as $O(N)$ vs $O(N\log N)$), but provides a smaller benefit.
Given a fixed time budget or fixed latency goal, is it better to repeatedly crack, sort, or mix the two approaches together.
We address this question with a family of scoring functions $\texttt{score}_\theta$, parameterized by a threshold value $\theta$ as follows:
\resizebox{\columnwidth}{!}{
$\texttt{score$_\theta$($\TransformInstance$,$\ArrayAtom(\ArrayOf{r_1\ldots r_N})$)} = \begin{cases}
N & \textbf{if } \TransformInstance = \SortTransform \textbf{ and } N < \theta\\
N & \textbf{if } \TransformInstance = \CrackTransform \textbf{ and } N < \theta\\
0 & \textbf{otherwise}
\end{cases}$
}
$\ArrayAtom$s smaller than the threshold are sorted, while those larger are cracked.
Larger instances are preferred over smaller.
All other instances are ignored.
Once all \ArrayAtom{}s are sorted, the resulting \SortedAtom{}s are iteratively \MergeTransform{}ed, ultimately leaving behind a single \SortedAtom{}.

This is one example of a parameterized policy, a reorganizational strategy that uses thresholds to guide its behavior.  
Once such a policy is defined, the next challenge is to select appropriate values for its parameters.

%% file: sections/model.tex
% -*- root: ../main.tex -*-

A \IndexName's performance curve depends on its policy.
%With a changing workload environment the requirements from the JITD changes with time. 
As we may have a range of policies to choose from --- for example by varying policy parameters as mentioned above --- we want a way to evaluate the utility of a policy for a given workload.
Naively, we might do this by repeatedly evaluating the structure under each policy, but doing so can be expensive.
Instead, we next propose a performance model for \IndexNames, policies, and a lightweight simulator that approximates the performance of a policy over time.  
Our approach is to see each transformation as an overhead performed in exchange for improved query performance. 
Hence, our model is based on two measured characteristics of the \IndexName: The costs of accessing an instance, and the cost of applying a transform.
A separate driver program measures (1) the cost of each access path on each instance atom type, varying every parameter available, and (2) the cost of each case of every transform.

\begin{example}
As an illustrative example, we will use the Crack-or-Sort policy described above.
This policy makes use of the $\ArrayAtom$, $\SortedAtom$, and $\BTreeAtom$ atoms, as well as the $\CrackTransform$ and $\SortTransform$ transforms.  
For this policy we need to measure 5 factors.
{\small
\begin{center}
\begin{tabular}{lcc}
\textbf{Operation} & \textbf{Symbol} & \textbf{Scaling}\\ \hline
\texttt{Get($\ArrayAtom(\ArrayOf{r_1\ldots r_N})$)} & $\alpha(N)$ & $O(N)$ \\
\texttt{Get($\SortedAtom(\ArrayOf{r_1\ldots r_N}))$)} & $\beta(N)$ & $O(\log_2(N))$ \\
\texttt{Get($\BTreeAtom(\KeyInstance, C_1, C_2))$)} & $\gamma$ & $O(1)$ \\
$\CrackTransform(\ArrayAtom(\ArrayOf{r_1\ldots r_N}))$ & $\delta(N)$ & $O(N)$ \\
$\CrackTransform(\ArrayAtom(\ArrayOf{r_1\ldots r_N}))$ & $\nu(N)$ & $O(n\log_2(n))$
\end{tabular}
\end{center}
}
The driver program fits each of the five functions by conducting multiple timing experiments, varying the size of $N$ where applicable.
\end{example}

The simulator mirrors the behavior of the full \IndexName, but uses a lighter-weight version of the \CogName grammar that does not store actual data:
\begin{eqnarray*}
\CogType^{\ell} &=&       \ArrayAtom(\mathbb N_0)
         \;\;|\;\; \SortedAtom(\mathbb N_0)\\
         &|&       \ConcatAtom(\CogType^{\ell}, \CogType^{\ell})
         \;\;|\;\; \BTreeAtom(\CogType^{\ell}, \CogType^{\ell})
\end{eqnarray*}
The simulator iteratively simulates applying transforms to instances expressed in $\CogType^{\ell}$ according to the policy being simulated.  
After each transform, the simulator uses the measured cost of the transform to estimate the cumulative time spent reorganizing the index.
The simulator captures multiple performance metrics $\texttt{metric} : \CogType \rightarrow \mathbb R$.

\begin{example}
Continuing the example, one useful metric is the read latency for a uniformly distributed read workload on a Crack-or-Sort index.
{\footnotesize
$$
\texttt{latency}(C) = \begin{cases}
\alpha(N) 
  & \textbf{if } C = \ArrayAtom(N)\\
\beta(N) 
  & \textbf{if } C = \ArrayAtom(N)\\
\gamma + \frac{|C_1|}{|C|}\texttt{latency}(C_1) \\
\phantom{\gamma} + \frac{|C_2|}{|C|}\texttt{latency}(C_2) 
  & \textbf{if } C = \BTreeAtom(C_1, C_2)\\
\end{cases}$$
}
where $|C|$ is the sum of sizes of $\ArrayAtom$s and $\SortedAtom$s in $\DescendantsOfCog{C}$.  
\end{example}

The simulator produces a sequence of status intervals: periods during which index performance is fixed, prior to the pointer swap after the next transform is computed.
A user-provided utility function aggregates these intervals to provide a final utility score for the entire policy.
Given a finite set of policies, the optimizer tries each in turn and selects the one that best optimizes the utility function.
Given a parameterized policy, the optimizer instead uses gradient descent.

\begin{example}
Examples of utility functions for Crack-or-Sort include: (1) Minimize time spent with more than $\theta$ \texttt{Get()} latency, (2) Maximize throughput for $N$ seconds, (3) Minimize runtime of $N$ queries.
% $$\text{argmin}_{N}\left( \sum_{\textbf{event} : \textbf{event}.\text{latency} > \theta} \textbf{event}.\text{end} - \textbf{event}.\text{start}  \right)$$
\end{example}

%% file: sections/discussion.tex
% -*- root: ../main.tex -*-

Ideally, we would like \CogName to be expressive enough to encode the instantaneous state of any data structure.
Infinite generality is obviously out of scope for this paper.
However we now take a moment to assess exactly what index data structure design patterns are supported in a \IndexName.

As a point of reference we use a taxonomy of data structures proposed as part of the Data Calculator~\cite{DBLP:conf/sigmod/IdreosZHKG18}.  
The data calculator taxonomy identifies 22 design primitives, each with a domain of between 2 and 7 possible values.
Each of the roughly $10^{18}$ valid points in this 22-dimensional space describes one possible index structure.
To the best of our knowledge, this represents the most comprehensive a survey of the space of possible index structures developed to date.

The data calculator taxonomy views index structures through the general abstraction of a tree with inner nodes and leaf nodes.
This abstraction is sometimes used loosely: A hash table of size N, for example, is realized as as a tree with precisely one inner-node and N leaf nodes.
Each of the taxonomy's design primitives captures one set of mutually exclusive characteristics of the nodes of this tree and how they are translated to a physical layout.

\newcommand{\Yes}{\textcolor{green!60!black!60}{\ding{108}}}
\newcommand{\Part}{\textcolor{blue!50!black!50}{\ding{119}}}
\newcommand{\No}{\textcolor{red!30!blue!30!black!40}{\ding{109}}}
\newcommand{\Imm}{\textcolor{red}{\ding{55}}}

\begin{figure*}
\centering
  \begin{tabular}{rl|c|l}
    \cline{1-4}\\[-3mm]
    \#&\textbf{Data Calculator Primitive} & \IndexName & \textbf{Note}\\%& Comments\\
    \cline{2-4}

    %\multirow{2}{*}{$\Big\}$ foo}
    {\tiny 1}  & Key retention                & \Part & No partial keys\\
    {\tiny 2}  & Value retention              & \No   & \\
    {\tiny 3}  & Key order                    & \Part & K-ary orders unsupported\\  % 2/3
    {\tiny 4}  & Key-Value layout             & \No   & No columnar layouts\\
    {\tiny 5}  & Intra-node access            & \Yes  & \\ %4/4
    {\tiny 6}  & Utilization                  & \Imm  & \\
    {\tiny 7}  & Bloom filters                & \No   & \\
    {\tiny 8}  & Zone map filters             & \Part & Implicit via \BTreeAtom\\
    {\tiny 9}  & Filter memory layout         & \No   & Requires filters (7,8)\\
    {\tiny 10}  & Fanout/Radix                & \No   & Limited to 2-way fanout\\
    {\tiny 11}  & Key Partitioning            & \Yes  & \\ %7/7
    {\tiny 12} & Sub-block capacity           & \Imm  & \\
    {\tiny 13} & Immediate node links         & \No   & Simulated by iterator impl.\\
    {\tiny 14} & Skip-node links              & \No   & \\
    {\tiny 15} & Area links                   & \No   & Simulated by iterator impl.\\
    {\tiny 16} & Sub-block physical location. & \No   & Only pointed supported\\
    {\tiny 17} & Sub-block physical layout.   & \Part / \Imm & Realized by merge rewrite\\ %1/3, but special case
    {\tiny 18} & Sub-block homogeneous        & \Yes  \\ %2/2
    {\tiny 19} & Sub-block consolidation      & \Yes  & Depends on policy\\ %2/2
    {\tiny 20} & Sub-block instantiation      & \Yes  & Depends on policy\\ %2/2
    {\tiny 21} & Sub-block link layout        & \No   & Requires links (13,14,15)\\
    {\tiny 22} & Recursion allowed            & \Yes  \\ %2/2
     % Scans                        & \Yes  \\%&  \IndexName / \IndexNames supports point look ups as well as range scans as describer in section \ref{sec:overview}\\
     % Sorted Search                & \Yes  \\%& All access techniques support binary search.\\
     % Hash Probe                   & \No   \\
     % Bloom Filter Probe           & \No   \\
     % Sort                         & \Yes  \\%& The sort policy as decsribed in \ref{sec:transformations} works on quick sort.\\
     % Random Memory Access         &      \\
     % Batched Memory Access        &      \\
     % Unordered Batch Write        &      \\
     % Ordered Batch Write          &      \\
     % Scattered Batch Write        &      \\
    \cline{1-4}
  \end{tabular}\\[1mm]
  \Yes: Full Support \hspace{5mm} \Part: Partial Support \hspace{5mm} \No: Support Possible\\
  \Imm: Not applicable to immutable data structures
  \caption{\IndexName support for the DC Taxonomy~\cite{DBLP:conf/sigmod/IdreosZHKG18}}
  \label{tab:dataprim}
  \trimfigurespacing
\end{figure*}

\Cref{tab:dataprim} classifies each of the design primitives as (1) Fully supported by \IndexName if it generalizes the entire domain, (2) Partially supported by \IndexName if it supports more than one element of the domain, or (3) Not supported otherwise.
We further subdivide this latter category in terms of whether support is feasible or not.
In general, the only design primitives that \IndexName can not generalize are related to mutability, since \IndexName's (semi-)immutability is crucial for concurrency, which is in turn required for optimization in the background.

\IndexName completely generalizes 7 of the remaining 22 primitives.
We first explain these primitives and how \IndexNames generalize them.
Then, we propose three extensions that, although beyond the scope of this paper, would fully generalize the final 14 primitives.  
For each, we briefly discuss the extension and summarize the challenges of realizing it.

\tinysection{Key retention (1)} This primitive expresses whether inner nodes store keys (in whole or in part), mirroring the choice between \ConcatAtom{} and \BTreeAtom.

\tinysection{Intra-node access (5)} This primitive expresses whether nodes (inner or child) allow direct access to specific children or whether they require a full scan, mirroring the distinction between \CogName nodes with and without semantic constraints.

\tinysection{Key partitioning (9)} This primitive expresses how newly added values are partitioned.  Examples include by key range (as in a B+Tree) or temporally (as in a log structured merge tree~\cite{DBLP:journals/acta/ONeilCGO96}).  Although a \IndexName only allows one form of insertion, policies can converge to the full range of states permitted for this primitive.

\tinysection{Sub-block homogeneous (18)} This primitive expresses whether all inner nodes are homogeneous or not.  %\IndexNames support both.

\tinysection{Sub-block consolidation/instantiation (19/20)} These primitives express how and when organization happens, as would be determined by a \IndexName's policy.

\tinysection{Recursion allowed (22)} This primitive expresses whether inner nodes form a bounded depth tree, a general tree, or a ``tree'' with a single node at the root.  \IndexNames support all three.

\subsection{Supporting New \CogName Atoms}

Five of the remaining primitives can be generalized by the addition of three new atoms to \CogName.
First, we would need a generalization of \BTreeAtom{} atoms capable of using partial keys as in a Trie (primitive 1), or hash values (primitive 3)
Second, a unary \textbf{Filter} atom that imposes a constraint on the records below it could implement both boom filters (primitives 7,9) and zone maps (primitives 8,9).  
These two atoms are conceptually straightforward, but introduce new transforms and increase the complexity of the search for effective policies.

The remaining challenge is support for columnar/hybrid layouts (primitive 4).
Columnar layouts increase the complexity of the formalism by requiring multiple record types and support for joining records.
Accordingly, we posit that a binary \textbf{Join} atom, representing the collection of records obtained by joining its two children could efficiently capture the semantics of columnar (and hybrid) layouts. 

\subsection{Atom Synthesis}

Five of the remaining primitives express various tactics for removing pointers by inlining groups of nodes into contiguous regions of memory.  
These primitives can be generalized by the addition of a form of atom synthesis, where new atoms are formed by merging existing atoms.  
Consider the Linked List of \Cref{ex:linkedlist}.
Despite the syntactic restriction over \CogName, a single linked list element must consist of two nodes (a \ConcatAtom{} and a (single-record) \ArrayAtom{}), and an unnecessary pointer de-reference is incurred on every lookup.
Assume that we could define a new node type: A linked list element 
($\textbf{Link}(\mathcal R, C)$) consisting of a record and a forward pointer.
Because this node type is defined in terms of existing node types, it would be possible to automatically synthesize new transformations for it from existing transformations, and existing performance models could likewise be adapted.

Atom synthesis could be used to create inner nodes that store values (primitive 2), increase the fanout of \ConcatAtom{} and \BTreeAtom{} nodes (primitive 10), inline nodes (primitive 16), and provide finer-grained control over physical layout of data (primitive 17).

\subsection{Links / DAG support}

The final four remaining properties (13, 14, 15, and 20) express a variety of forms of link between inner and leaf nodes. 
Including such links turns the resulting structure into a directed acyclic graph (DAG).
In principle, it should be possible to generalize transforms for arbitrary DAGs rather than just trees as we discuss in this paper.  
Such a generalization would require additional transforms that create/maintain the non-local links and more robust garbage collection.

%% file: sections/evaluation.tex
% -*- root: ../main.tex -*-
%!TEX root = ../main.tex

We next evaluate the performance of \IndexNames in comparison to other commonly used data structures.
Our results show that: 
(1) In the longer term, \IndexNames have minimal overheads relative to standard in-memory data structures;
(2) The \IndexName policy simulator reliably models the behavior of a \IndexName; 
(3) In the short term, \IndexNames can out-perform standard in-memory data structures; 
(4) Concurrency introduces minimal overheads; and
(5) \IndexNames scale well with data, both in their access costs and their organizational costs.

\subsection{Experimental setup}

All experiments were run on a 2$\times$6-core 2.5 GHz Intel Xeon server with 198 GB of RAM and running Ubuntu 16.04 LTS.  Experimental code was written in C++ and compiled with GNU C++ 5.4.0.  
Each element in the data set is a pair of key and value, each an 8-Byte integer.
Unless otherwise noted, we use a data size of up to a maximum of $10^9$ records (16GB) with keys generated uniformly at random.
To mitigate experimental noise, we use srand() with an arbitrary but consistent value for all data generation.
To put our performance numbers into context, we compare against 
(1) \textbf{R/B Tree}: the C++ standard-template library (STL) \texttt{map} implementation (a classical red-black tree),
(2) \textbf{HashTable} the C++ standard-template library (STL) \texttt{unordered-map} implementation (a hash table), and
(3) \textbf{BTree} a publicly available implementation of b-trees\footnote{\url{https://github.com/JGRennison/cpp-btree}}.
For all three, we used the \texttt{find()} method for point lookups and \texttt{lower\_bound()}/\texttt{++} (where available) for range-scans.
For point lookups, we selected the target key uniformly at random\footnote{We also tested a heavy-hitter workload that queried for 30\% of the keyspace 80\% of the time, but found no significant differences between the workloads.}.
For range scans, we selected a start value uniformly at random and the end value to visit approximately 1000 records.
Except where noted, access times are the average of 1000 point lookups or 50 range scans.

We specifically evaluated \IndexNames using the Crack-Sort-Merge family of policies described in Section~\ref{sec:crack-or-sort}, varying the crack threshold over $10^6$, $10^7$, $10^8$, and $10^9$ records.  When there are exactly $10^9$ records, this last policy simply sorts the entire input in one step. 
For point lookups we use the \texttt{get()} access path, and for range scans we use the \texttt{ordered\_iterator()} access path. 
By default, we measure \IndexName read performance through a synchronous (i.e., with the worker thread paused) microbenchmark.
We contrast synchronous and asynchronous performance in \Cref{sec:threaded}.

Synchronous read performance was measured through a sequence of trials, each with a progressively larger number of transforms (i.e., a progressively larger fragment of the policy's trace) applied to the \IndexName.  
We measured total time to apply the trace fragment (including the cost of selecting which transforms to apply) before measuring access latencies.
For concurrent read performance a client thread measured access latency approximately once per second.

\begin{figure*}
\centering
  \begin{subfigure}{0.49\textwidth}
    \centering
    \includegraphics[width=0.9\textwidth]{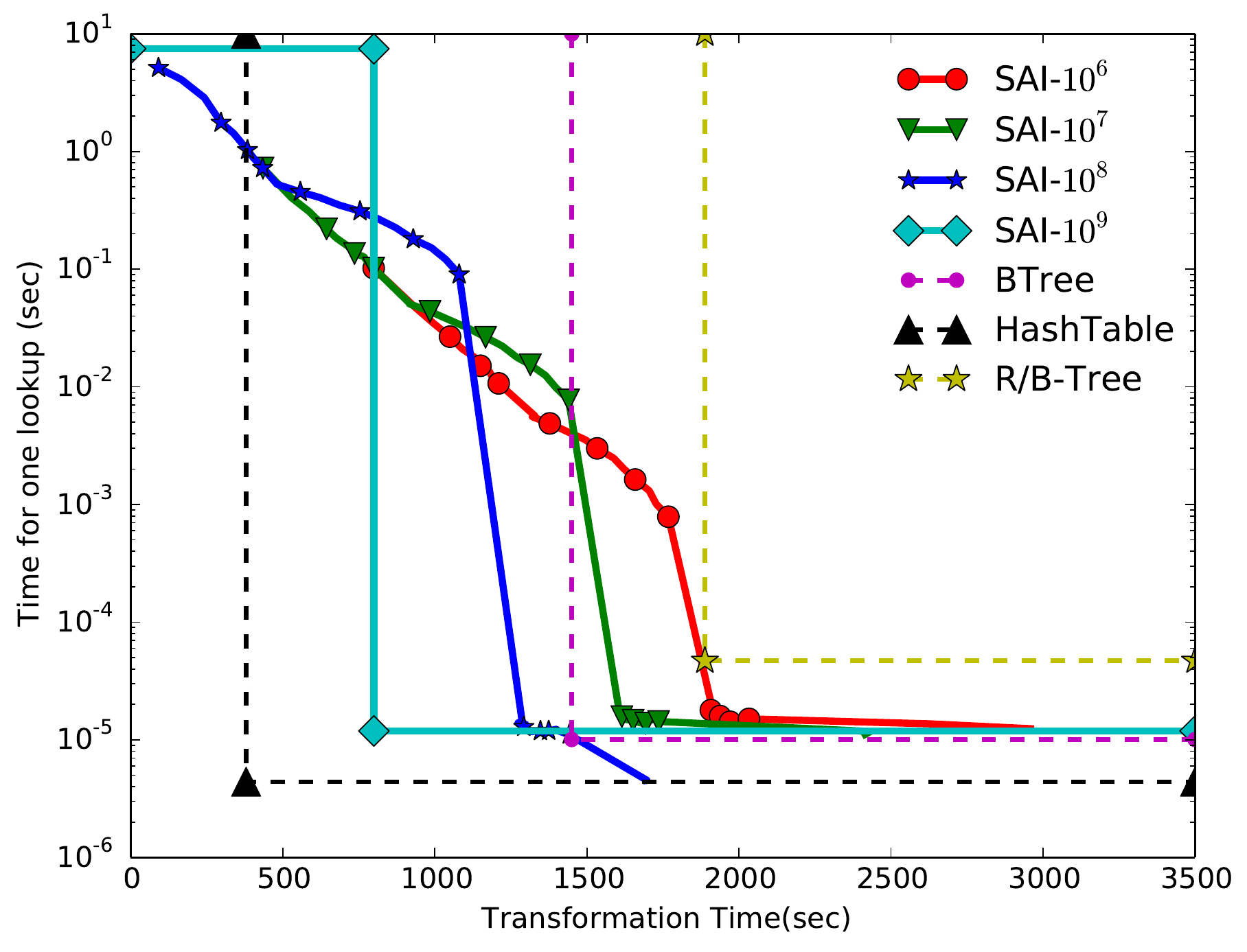}
    \caption{Point Lookups}
    \label{fig:StagedTransformslookupVspolicytime}
  \end{subfigure}
  \begin{subfigure}{0.49\textwidth}
    \centering
    \includegraphics[width=0.9\columnwidth]{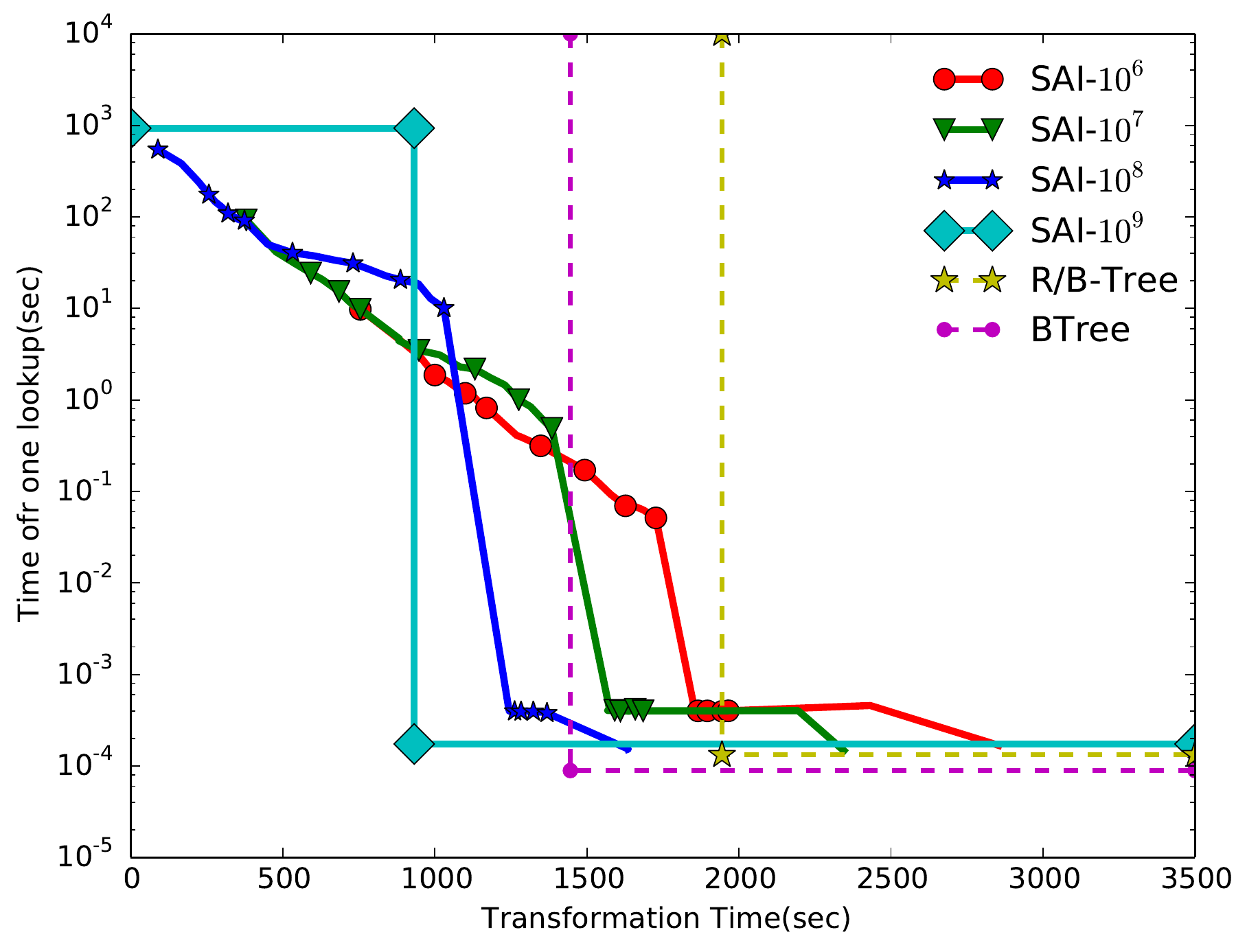}
    \caption{Range Scans}
    \label{fig:StagedTransformslookupVspolicytimerange-scans}
  \end{subfigure}
  \vspace{-3mm}
  \caption{Performance improvement over time as each \IndexName is organized}
  \trimfigurespacing
\end{figure*}

\subsection{Cost vs Benefit Over Time}

Our first set of experiments mirrors \Cref{fig:scalingComparison}, tracking the synchronous performance of point lookups and range scans over time.
The results are shown in \Cref{fig:StagedTransformslookupVspolicytime} and \Cref{fig:StagedTransformslookupVspolicytimerange-scans}
The x-axis shows time elapsed, while the y-axis shows index access latency at that point in time. 
In both sets of experiments, we include access latencies and setup time for the R/B-Tree (yellow star), the HashTable (black triangle), and the BTree (pink circles)
We treat the cost of accessing an incomplete data structure as infinite, stepping down to the structure's normal access costs once it is complete.

In general, lower crack thresholds achieve faster upfront performance by sacrificing long-term performance.
A crack threshold of $10^6$ (approximately $\frac{1}{10^5}$ cracked partitions) takes approximately twice as long to reach convergence as a threshold of $10^9$ (sort everything upfront)

Unsurprisingly, for point lookups the Hash Table has the best overall performance curve.  However, even it needs upwards of 6 minutes worth of data loading before it is ready.   
By comparison, a \IndexName starts off with a 10 second response time, and has dropped to under 3 seconds by the 3 minute mark.
The BTree significantly outperforms the R/B-Tree on both loading and point lookup cost, but still takes nearly 25 minutes to fully load.  By that point the Threshold$10^8$ policy \IndexName has already been serving point lookups with a comparable latency (after its sort phase) for nearly 5 minutes.  
Note that lower crack thresholds have a slightly slower peak performance than higher ones before their merge phase
This is a consequence of deeper tree structures and the indirection resulting from handles.% and we quantify it below.
The performance at convergence of the $10^8$ threshold point scan trial is surprising, as it suggests binary search is as fast as a hash lookup.  
We suspect this due to lucky cache hits, but have not yet been able to confirm it.

\subsection{Simulated vs Actual Performance}
Figure~\ref{fig:TransformationTimeVsScanTimeEstimated} shows the result of using our simulator to predict the performance curves of Figure~\ref{fig:StagedTransformslookupVspolicytime}.  As can be seen, performance is comparable.  
Policy runtimes are replicated reliably, features like time to convergence and crossover are replicated virtually identically.

\begin{figure}
\centering
\includegraphics[width=0.9\columnwidth]{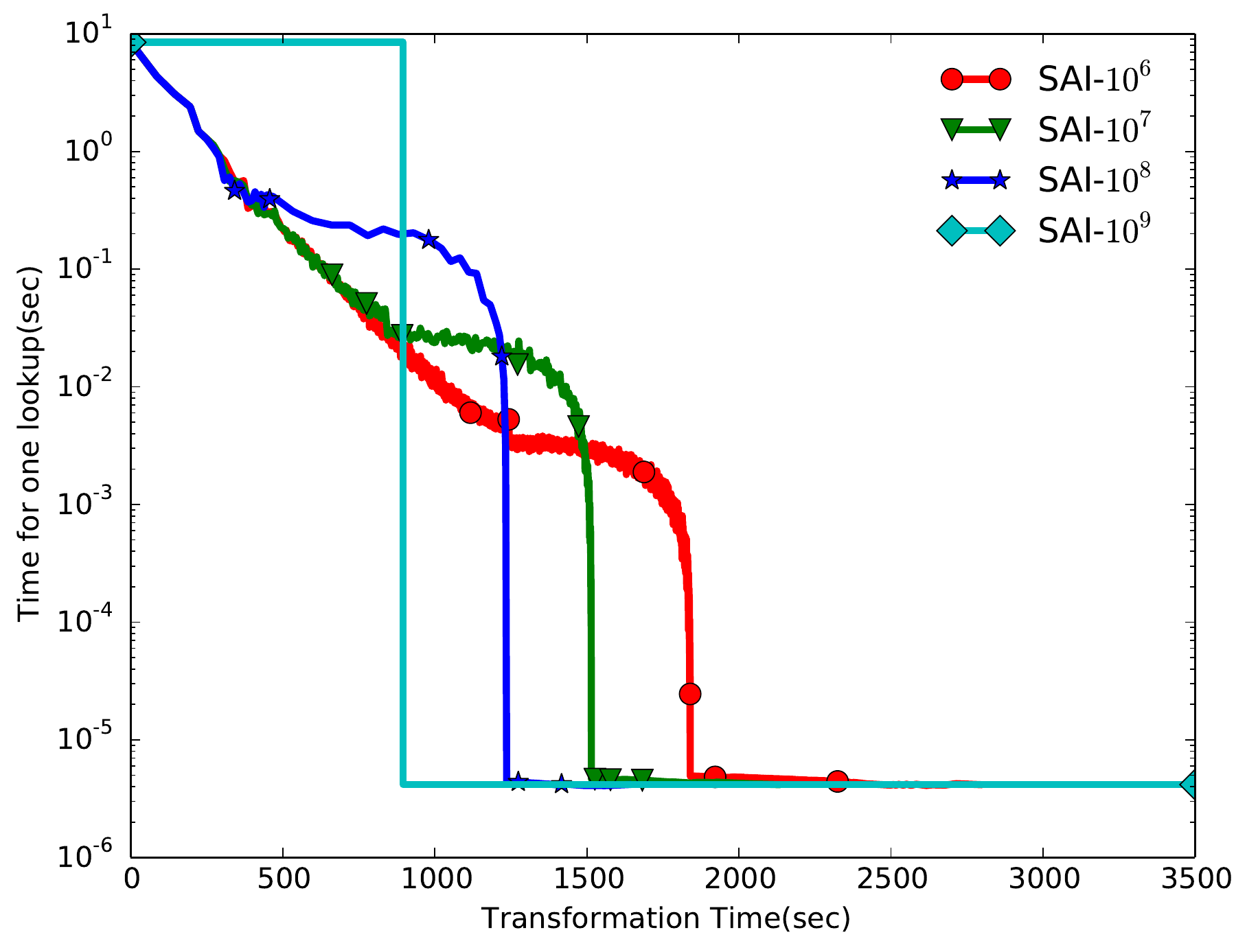}
\caption{Predicted performance using the simulator.}
\label{fig:TransformationTimeVsScanTimeEstimated}
\trimfigurespacing
\vspace{1mm}
\end{figure}

\subsection{Synchronous vs Concurrent}
\label{sec:threaded}
Figures~\ref{fig:threading6}, \ref{fig:threading7}, and \ref{fig:threading8} contrast the synchronous performance of the \IndexName with a more realistic concurrent workload.  
Performance during the crack phase is comparable, though admittedly with a higher variance.  
As expected, during the sort phase performance begins to bifurcate into fast-path accesses to already sorted arrays and slow-path scans over array nodes at the leaves.  

The time it takes the worker to converge is largely unaffected by the introduction of concurrency.
However, as the structure begins to converge, we see a constant $100\mu s$ overhead compared to synchronous access.  
We also note periodic $100ms$ bursts of latency during the sort phases of all trials.  
We believe these are caused when the worker thread pointer-swaps in a new array during the merge phase, as the entire newly created array is cold for the client thread.

\begin{figure*}
\centering
  \begin{subfigure}{0.33\textwidth}
    \centering
    \includegraphics[width=\textwidth]{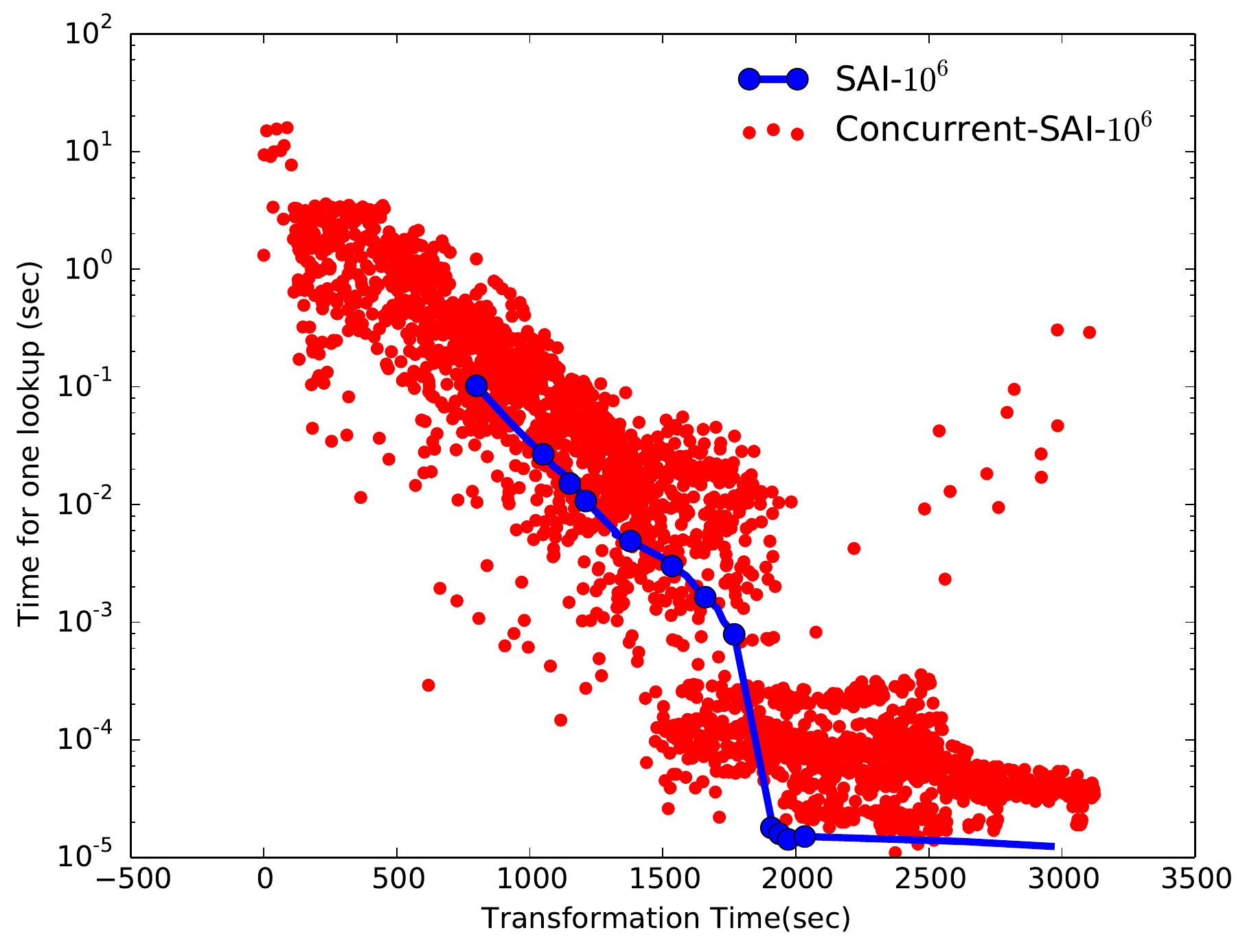}
    \caption{Crack Threshold = $10^6$}
    \label{fig:threading6}
  \end{subfigure}
  \begin{subfigure}{0.33\textwidth}
    \centering
    \includegraphics[width=\textwidth]{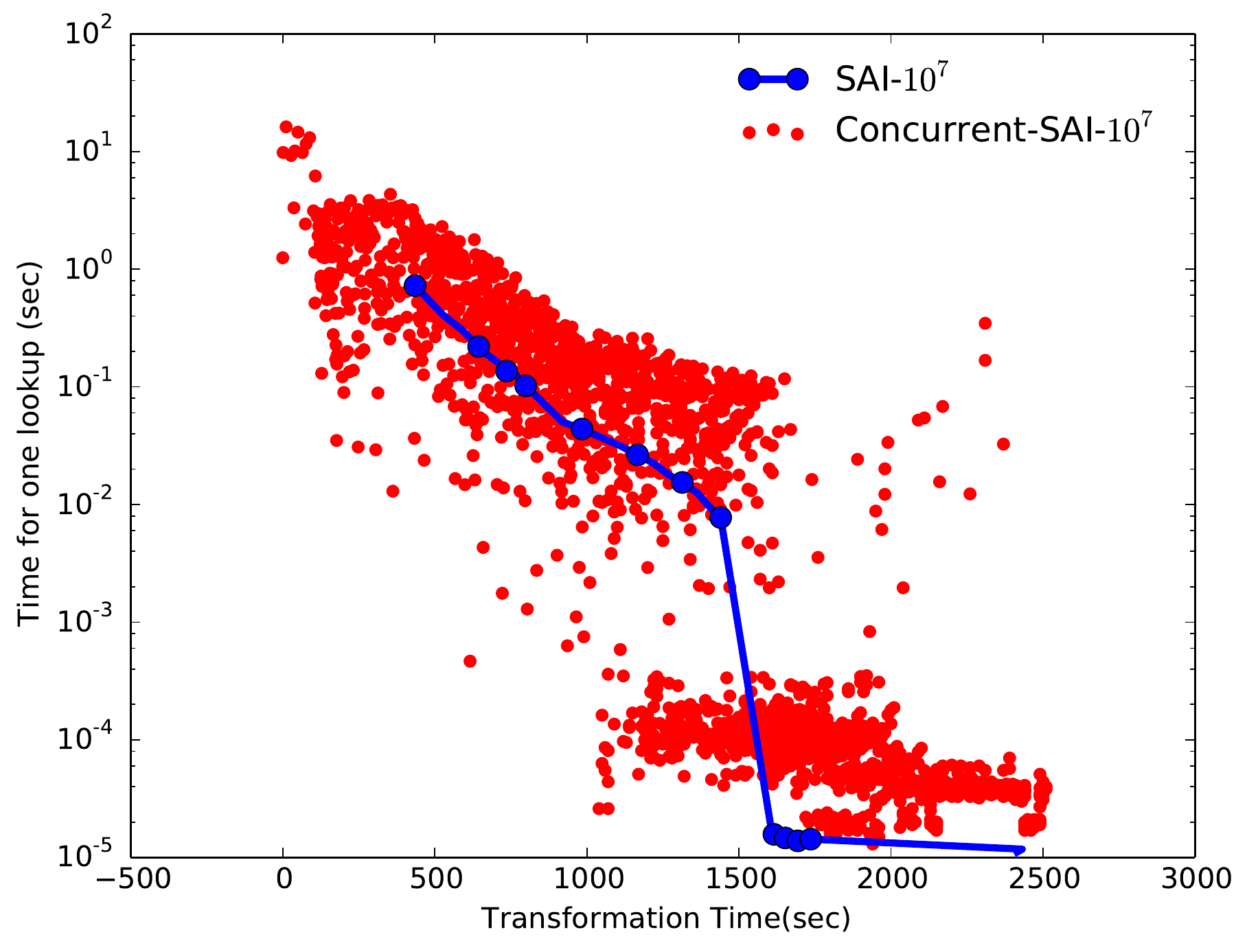}
    \caption{Crack Threshold = $10^7$}
    \label{fig:threading7}
  \end{subfigure}
  \begin{subfigure}{0.33\textwidth}
    \centering
    \includegraphics[width=\textwidth]{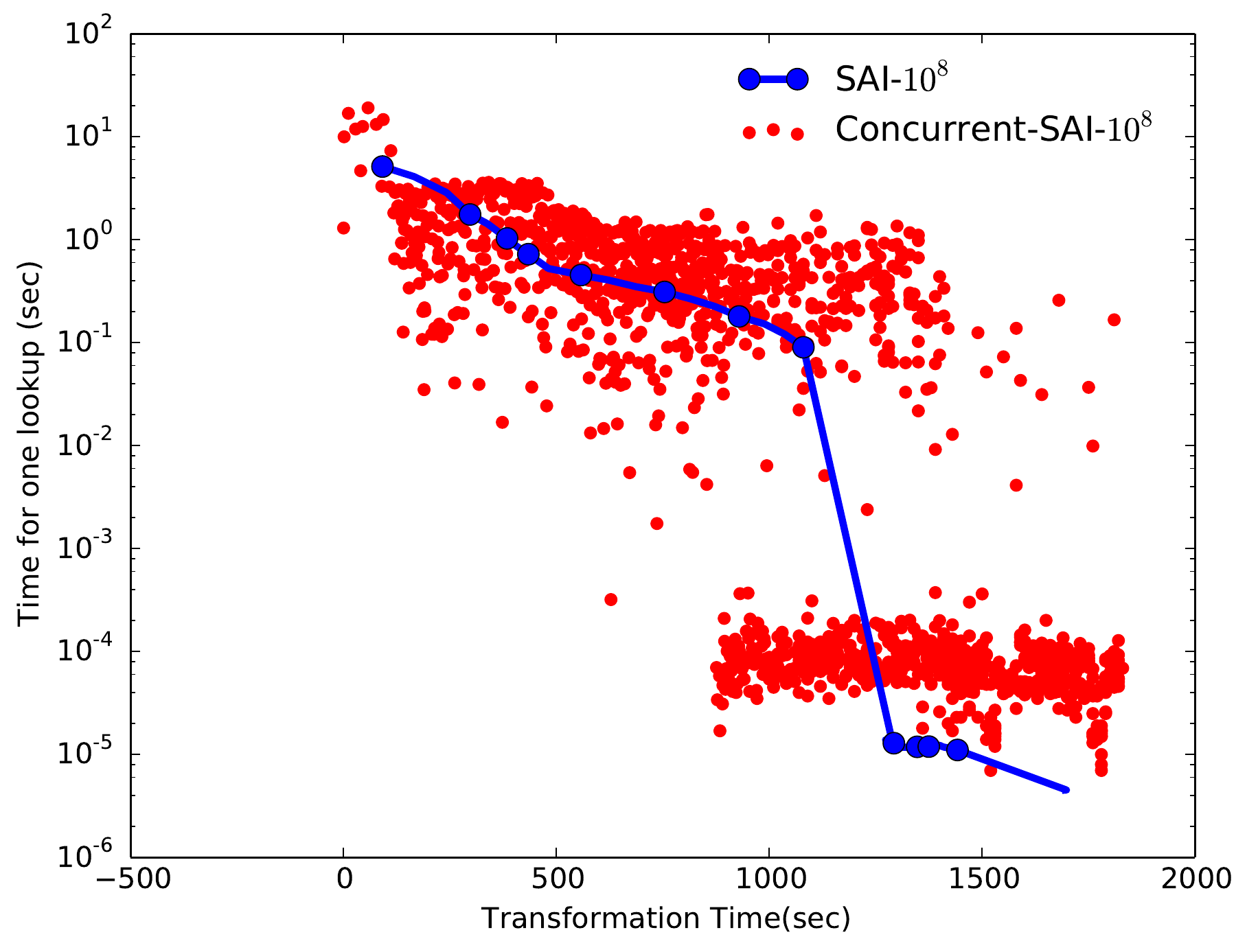}
    \caption{Crack Threshold = $10^8$}
    \label{fig:threading8}
  \end{subfigure}
  % \begin{subfigure}{0.33\textwidth}
  %   \centering
  %   \includegraphics[width=0.8\textwidth]{plots/Threading-9.pdf}
  %   \caption{Crack Threshold = $10^8$}
  %   \label{fig:threading9}
  % \end{subfigure}
\vspace{-2mm}
\caption{Synchronous vs Concurrent performance of the \IndexName on point lookups.}
\trimfigurespacing
\end{figure*}

\subsection{Short-Term Benefits for interactive workloads}
One of the primary benefits of \IndexNames is that they can provide significantly better performance during the transition period.
This is particularly useful in interactive settings where users pose tasks comparatively slowly.  
We next consider such a hypothetical scenario where a data file is loaded and each data structure is given a short period of time (5 seconds) to prepare.  
In these experiments, we use a cracking threshold of $10^5$ (our worst case), and vary the size of the data set from $10^6$ records (16MB) to $10^9$ records (16GB). 
The lookup time is the time until an answer is produced: the cost of a point lookup for the \IndexName.
The baseline data structures are accessible only once fully loaded, so we model the user waiting until the structure is ready before doing a point lookup.
Up through $10^7$ records, the \texttt{unordered\_map} completes loading within 5 seconds.  In every other case, the \IndexName is able to produce a response orders of magnitude faster.

\begin{figure}
\centering
\includegraphics[width=0.8\columnwidth]{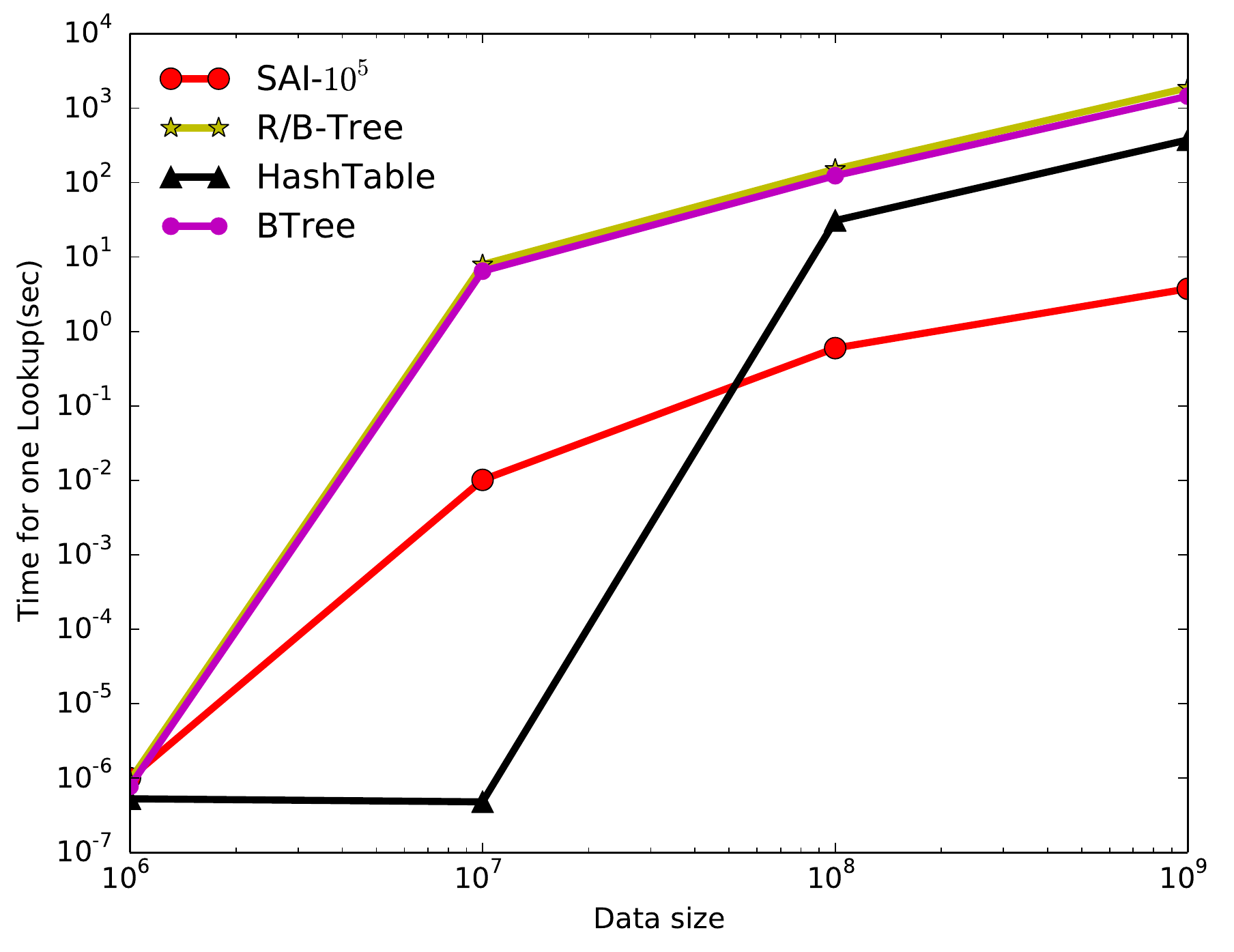}
\caption{Point lookup latency relative to data size.}
\label{fig:DataSizeVsScanTime}
\trimfigurespacing
\end{figure}
% \begin{figure}
% \centering
% \includegraphics[width=70mm]{plots/DataSizeVsScanTime_HeavyHitter.pdf}
% \caption{For the Heavy Hitter workload each point on the graph indicates the change in scan time with the increase in data size.}
% \label{fig:DataSizeVsScanTimeHeavy}
% \end{figure}

\subsection{DataSize Vs TransformTime}
Figure~\ref{fig:datasizevsloadtime} illustrates the scalability of \IndexName from the perspective of data loading.  
As before, we vary the size of the data set and use the time taken to load a comparable amount of data into the base data structures.
Note that data is accessible virtually immediately after being loaded into a \IndexName. 
We measure the cost for the \IndexName to reach convergence.
The performance of the \IndexName and the other data structures both scale linearly with the data size (note the log scale).
%\textcolor{red}{EXPLAIN THE DIFFERENCE.}

\begin{figure}
\centering
\includegraphics[width=0.8\columnwidth]{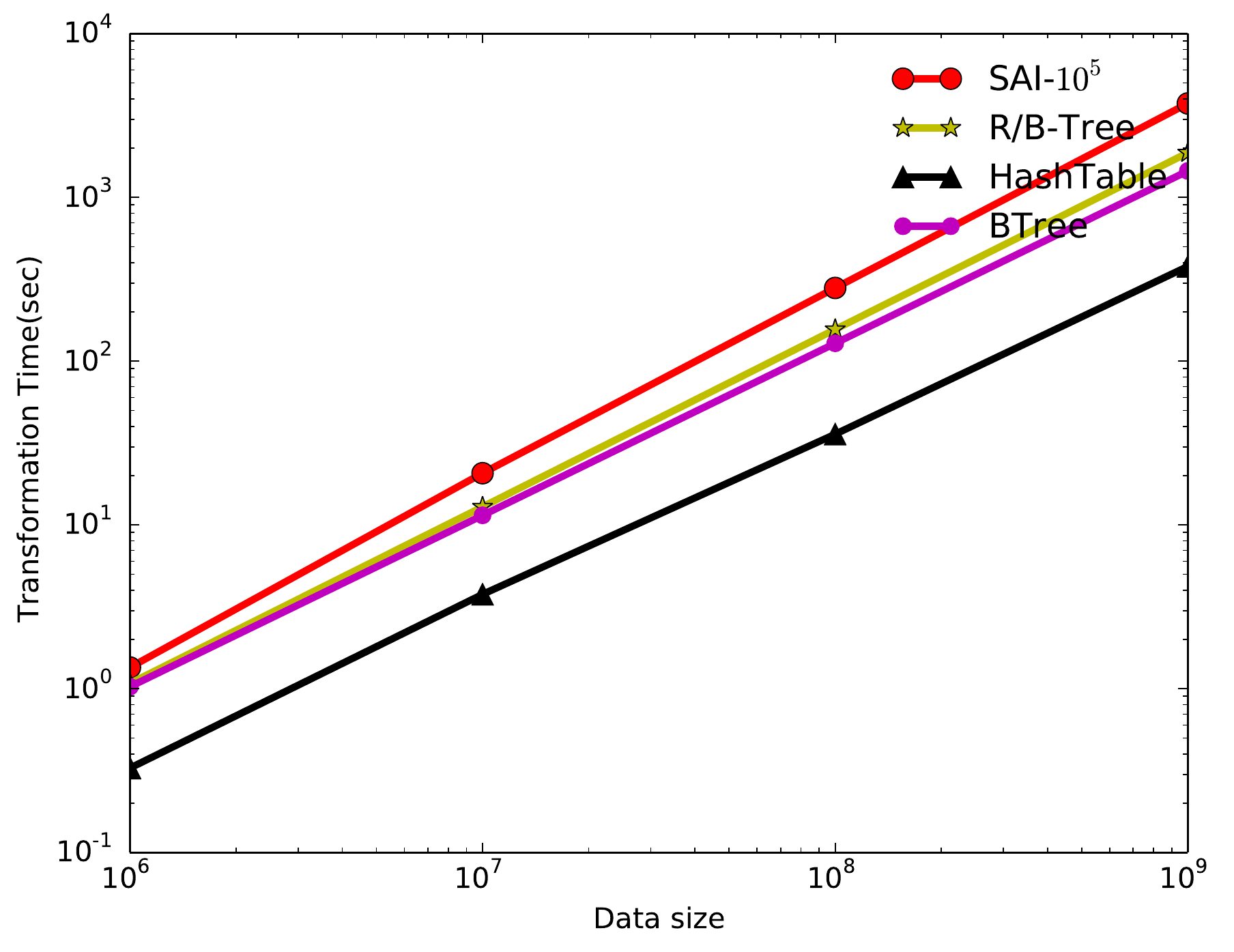}
\caption{The time required to load and fully organize a data set relative to data size}
\label{fig:datasizevsloadtime}
\trimfigurespacing
\end{figure}

\subsection{CrackThreshold Vs ScanTime}
Figure~\ref{fig:CrackThresholdVsScanTime} explores the effects of the crack threshold on performance at convergence of the Crack-Sort policy. The Merge Policy was excluded from testing as at convergence it would lead to one huge sorted array of size $10^9$ irrespective of the crack threshold.
In these set of experiments the crack threshold for cracking an array in the JITD structure was varied from $10^6$ to $10^9$.
For each, we performed one thousand point scans, measuring the total time and computing the average cost per scan.
This figure shows the overhead from handles --- at a crack threshold of $10^9$, the entire array is sorted in a single step.
As the crack threshold grows by a factor of 10, the depth of the tree increases by roughly a factor of three, necessitating approximately 3 additional random accesses via handles rather than directly on a sorted array, and as shown in the graph, increasing access time by roughly 1 $\mu$s.  

\begin{figure}
\centering
\includegraphics[width=0.9\columnwidth]{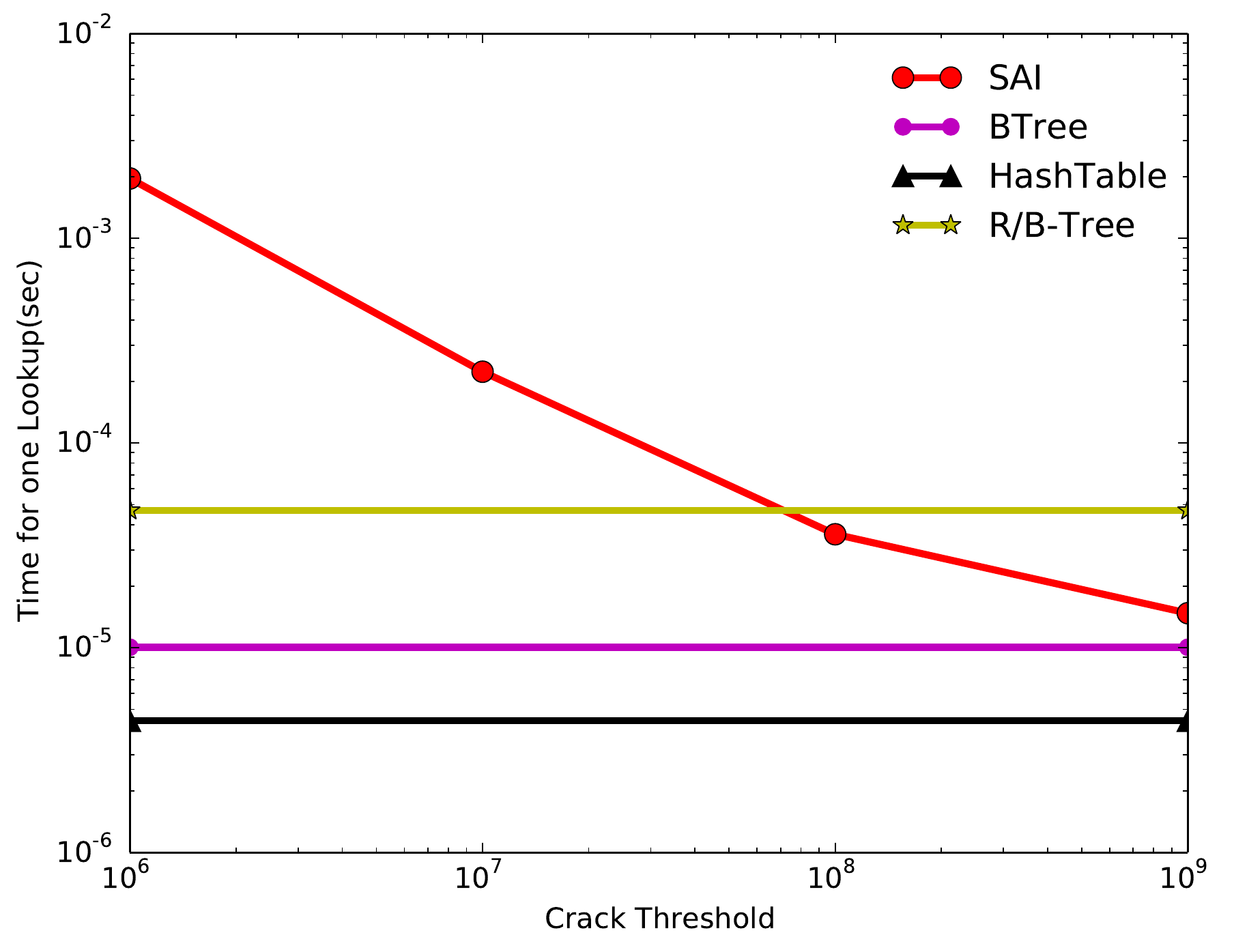}
\caption{For the Uniform workload each point on the graph indicates the change in scan time for different crack thresholds.}
\label{fig:CrackThresholdVsScanTime}
\trimfigurespacing
\end{figure} 
% \begin{figure}
% \centering
% \includegraphics[width=70mm]{plots/DataSizeVsPolicyTime_HeavyHitter.pdf}
% \caption{The graph indicates the effect of data size on the policy transformation time for a heavy hitter workload}
% \label{fig:TransformationTimeVsDataSizehh}
% \end{figure} 
% \begin{figure}
% \centering
% \includegraphics[width=70mm]{plots/CrackThresholdVsScanTime_HeavyHitter.pdf}
% \caption{For the Heavy Hitter workload each point on the graph indicates the change in scan time for different crack thresholds.}
% \label{fig:CrackThresholdVsScanTimeHeavy}
% \end{figure}

%% file: sections/related_work.tex
% -*- root: ../main.tex -*-
\IndexNames specifically extend work by Kennedy and Ziarek on Just-in-Time Data Structures~\cite{DBLP:conf/cidr/KennedyZ15} with a framework for defining policies, tools for optimizing across families of policies, and a runtime that supports optimization in the background rather than as part of queries.
Most notably, this enables efficient dynamic data reorganization as an ongoing process rather than as an inline, blocking part of query execution.

Our goal is also spiritually similar to The Data Calculator~\cite{DBLP:conf/sigmod/IdreosZHKG18}.
Like our policy optimizer, it searches through a large space of index design choices for one suitable for a target workload.
However, in contrast to \IndexNames, this search happens once at compile time and explores mostly homogeneous structures.
In principle, the two approaches could be combined, using the Data Calculator to identify optimal structures for each workload and using \IndexNames to migrate between structures as the workload changes.

Also related is a recently proposed form of ``Resumable'' Index Construction~\cite{DBLP:journals/pvldb/AntonopoulosKTU17}.  
The primary challenge addressed by this work is ensuring that updates arriving after index construction begins are properly reflected in the index.
While we solve this problem (semi-)functional data structures, the authors propose the use of temporary buffers.

\tinysection{Adaptive Indexing}
\IndexNames are a form of adaptive indexing~\cite{DBLP:conf/edbt/IdreosMG12,DBLP:conf/tpctc/GraefeIKM10}, an approach to indexing that re-uses work done to answer queries to improve index organization.  
Examples of adaptive indexes include Cracker Indexes~\cite{DBLP:conf/cidr/IdreosKM07,DBLP:conf/sigmod/IdreosKM07}, Adaptive Merge Trees~\cite{DBLP:conf/edbt/GraefeK10}, SMIX~\cite{DBLP:conf/ssdbm/VoigtKL13}, and assorted hybrids thereof~\cite{DBLP:journals/pvldb/IdreosMKG11,DBLP:conf/cidr/KennedyZ15}.
Notably, a study by Schuhknecht et. al.~\cite{DBLP:journals/pvldb/SchuhknechtJD13} compares (among other things) the overheads of cracking to the costs of upfront indexing.  
Aiming to optimize overall runtime, upfront indexing begins to outperform cracker indexes after thousands to tens of thousands of queries.
By optimizing the index in the background, \IndexNames avoid the overheads of data reorganization \emph{as part of the query itself}.

\tinysection{Organization in the Background}
Unlike adaptive indexes, which inline organizational effort into normal database operations, several index structures are designed with background performance optimization in mind.
These begin with work in active databases~\cite{DBLP:books/mk/widomC96/WidomC96}, where reactions to database updates may be deferred until CPU cycles are available.
More recently, bLSM trees~\cite{DBLP:conf/sigmod/SearsR12} were proposed as a form of log-structured merge tree that coalesces partial indexes together in the background.
A wide range of systems including COLT~\cite{DBLP:conf/sigmod/SchnaitterAMP06}, OnlinePT~\cite{DBLP:conf/icde/BrunoC07}, and Peloton~\cite{DBLP:conf/cidr/PavloAALLMMMPQS17} use workload modeling to dynamically select, create, and destroy indexes, also in the background.

\tinysection{Self-Tuning Databases}
Database tuning advisors have existed for over two decades~\cite{DBLP:conf/sigmod/ChaudhuriN98,DBLP:conf/vldb/ChaudhuriN07}, automatically selecting indexes to match specific workloads.
However, with recent advances in machine learning technology, the area has seen significant recent activity, particularly in the context of index selection and design.
OtterTune~\cite{DBLP:conf/sigmod/AkenPGZ17} uses fine-grained workload modeling to predict opportunities for setting database tuning parameters, an approach complimentary to our own.%, and could be used to create new atom types from existing ones.

\tinysection{Generic Data Structure Models}
More spiritually similar to our work is The Data Calculator~\cite{DBLP:conf/sigmod/IdreosZHKG18}, which designs custom tree structures by searching through a space of dozens of parameters describing both tree and leaf nodes.
A similarly related effort uses small neural networks~\cite{DBLP:conf/sigmod/KraskaBCDP18} as a form of universal index structure by fitting a regression on the CDF of record keys in a sorted array.

% \tinysection{Specialized Data Structures}
% Notably, 
% CrimsonDB~\cite{}

%% file: sections/future_work.tex
% -*- root: ../main.tex -*-

In this paper, we introduced \IndexNames a type of in-memory index that can incrementally morph its performance characteristics to adapt to changing workloads.
To accomplish this, we formalized a composable organizational grammar (\CogName) and a simple algebra over it.
We introduced a range of equivalence- and structure-preserving rewrite rules called transforms that serve as the basis of organizational policies that guide the transition from one performance envelope to another.
We described a simulation framework that enables efficient optimization of policy parameters.
Finally, we demonstrated that a \IndexName can be implemented with minimal overhead relative to classical in-memory index structures.

Our work leaves open several challenges.  
We have already identified three specific challenges in \Cref{sec:discussion}: New atoms, Atom synthesis, and DAG support.
Addressing each of these challenges would allow \CogName to capture a wide range of data structure semantics.
There are also several key areas where performance tuning is possible:
First, our use of reference-counted pointers also presents a performance bottleneck for high-contention workloads --- we plan to explore more active garbage-collection strategies.
Second, Handles are an extremely conservative realization of semi-functional data structures.
As a result, \IndexNames are a factor of 2 slower at convergence than other tree-based indexes.  
We expect that this performance gap can be reduced or eliminated by identifying situations where Handles are unnecessary (e.g., at convergence).
A final open challenge is the use of statistics to guide rewrite rules, both detecting workload shifts to trigger policy shifts (e.g., as in Peloton), as well as identifying statistics-driven policies that naturally converge to optimal behaviors for dynamic workloads.

%% file: sections/proofs.tex
% -*- root: ../main.tex -*-

\section{Correctness of Example Transforms}
\label{apx:ExampleTransformCorrectness}

As a warm-up and an example of transform correctness, we next review each of the transforms given in Figure~\ref{fig:transformExamples} and prove the correctness of each.

\begin{proposition}[Identity is correct] 
\label{thm:identityIsCorrect}
Let $\IdentityTransform$ denote the identity transform $\IdentityTransform(C) = C$.  $\IdentityTransform$ is both equivalence preserving and structure preserving.
\end{proposition}

\begin{lemma}[$\SortTransform$ is correct] 
\label{thm:sortIsCorrect}
$\SortTransform$ is both equivalence preserving and structure preserving.
\end{lemma}

\begin{proof}
For any instance $C$ where $\TypeOfCogInstance{C} \neq \ArrayAtom$, correctness follows from \Cref{thm:identityIsCorrect}.

Otherwise $C = \ArrayAtom(\ArrayOf{r_1, \ldots, r_N})$, and consequently $\SortTransform(C) = \SortedAtom(\SortFunction(\ArrayOf{r_1, \ldots, r_N}))$.  
To show correctness we first need to prove that
{\small
$$\ContentsOfCog{\ArrayAtom(\ArrayOf{r_1, \ldots, r_N})} = 
  \ContentsOfCog{\SortedAtom(\SortFunction(\ArrayOf{r_1, \ldots, r_N}))}$$
}
Let the one-to-one (hence invertable) function $f : [1,N] \rightarrow [1,N]$ denote the transposition applied by $\SortFunction$.
{\small
\\[2mm]$\ContentsOfCog{\SortedAtom(\SortFunction(\ArrayOf{r_1, \ldots, r_N}))}$
\begin{eqnarray*}
&&= \ContentsOfCog{\SortedAtom(\ArrayOf{r_{f^{-1}(1)}, \ldots, r_{f^{-1}(N)}}))}\\
&&= \BagOf{r_{f^{-1}(1)}, \ldots, r_{f^{-1}(N)}} \\
&&= \BagOf{r_1, \ldots, r_N}\\
&&= \ContentsOfCog{\ArrayAtom(\ArrayOf{r_1, \ldots, r_N})}
\end{eqnarray*}
}
\noindent giving us equivalence preservation.  Structure preservation requires that $\ArrayOf{r_{f^{-1}(1)}, \ldots, r_{f^{-1}(N)}}$ be in sorted order, which it is by construction.  Thus, $\SortTransform$ is a correct transform.
\end{proof}

% \noindent \textbf{Example :} Given an array cog C consisting of elements \big\{6,2,8,2,5,1,3\big\}. Let T be a transform on C that sorts the elements of the cog.  Applying the transform T on C would result in a SortedArray cog C' containing elements of C in a sorted order. 

% \hspace*{2.18cm} C: Array(\big\{6,2,8,2,5,1,3\big\}).

% \hspace*{2.18cm} T: Sort Elements of Cog

% \hspace*{2.18cm} $C \Longrightarrow _T   C'$

% \hspace*{2.18cm} C':SortedArray(\big\{1,2,2,3,5,6,8\big\})

% Let the data from C be represented by a set X = \big\{6,2,8,2,5,1,3\big\} and the data from C' be represented by Y =\big\{1,2,2,3,5,6,8\big\}. On comparing the two sets we can say X = Y. The physical structure of C and C' are diffrent but they both represent the same logical content.Hence, D(C) = D(C'). The transform T was defined to sort the elements of the cog and it preserved the content. Hence T is a correct transform on C.

%%%%%%%

\begin{lemma}[$\UnsortTransform$ is correct] 
\label{thm:unsortIsCorrect}
$\UnsortTransform$ is both equivalence preserving and structure preserving.
\end{lemma}
\begin{proof}
For any instance $C$ where $\TypeOfCogInstance{C} \neq \SortedAtom$, correctness follows from \Cref{thm:identityIsCorrect}.

Otherwise $C = \SortedAtom(\ArrayOf{r_1 \ldots r_N})$ and we need to show first that $\ContentsOfCog{\SortedAtom(\ArrayOf{r_1 \ldots r_N})} = \ContentsOfCog{\ArrayAtom(\ArrayOf{r_1 \ldots r_N})}$.
The logical contents of both are $\BagOf{r_1 \ldots r_N}$, so we have equivalence.
Structure preservation is a given since any $\ArrayAtom$ instance is structurally correct.
\end{proof}

%%%%%%%%

\begin{lemma}[$\DivideTransform$ is correct] 
\label{thm:divideIsCorrect}
$\DivideTransform$ is both equivalence preserving and structure preserving.
\end{lemma}
\begin{proof}
For any instance $C$ where $\TypeOfCogInstance{C} \neq \ArrayAtom$, correctness follows from \Cref{thm:identityIsCorrect}.

Otherwise $C = \ArrayAtom(\ArrayOf{r_1 \ldots r_N})$ and we need to show first that 
{\footnotesize
\begin{multline*}
\ContentsOfCog{\ArrayAtom(\ArrayOf{r_1 \ldots r_N})} =\\
\ContentsOfCog{\ConcatAtom\left(
    \ArrayAtom(\ArrayOf{r_1 \ldots r_{\FloorOf{\frac{N}{2}}}}), 
    \ArrayAtom(\ArrayOf{r_{\FloorOf{\frac{N}{2}}+1} \ldots r_N})\right )}
\end{multline*}
}
Evaluating the right hand side of the equation recursively and simplifying, we have
\begin{eqnarray*}
&=& \BagOf{r_1 \ldots r_{\FloorOf{\frac{N}{2}}}} \uplus 
    \BagOf{r_{\FloorOf{\frac{N}{2}}+1} \ldots r_N}\\
&=& \BagOf{r_1 \ldots r_{\FloorOf{\frac{N}{2}}}, r_{\FloorOf{\frac{N}{2}}+1} \ldots r_N}\\
&=& \BagOf{r_1 \ldots r_N} \;\;=\;\;\ContentsOfCog{\ArrayAtom(\ArrayOf{r_1 \ldots r_N})}
\end{eqnarray*}
Hence we have equivalence preservation.  The $\ArrayAtom$ instances are always structurally correct and $\ConcatAtom$ instances are structurally correct if their children are, so we have structural preservation as well. Hence, $\DivideTransform$ is correct.
\end{proof}

%%%%%%%%

\begin{lemma}[$\CrackTransform$ is correct]
\label{thm:crackIsCorrect}
$\CrackTransform$ is both equivalence preserving and structure preserving.
\end{lemma}
\begin{proof}
For any instance $C$ where $\TypeOfCogInstance{C} \neq \ArrayAtom$, correctness follows from \Cref{thm:identityIsCorrect}.

Otherwise $C = \ArrayAtom(\ArrayOf{r_1 \ldots r_N})$ and we need to show first that 
{\footnotesize
\begin{equation*}\begin{split}
&\ContentsOfCog{\ArrayAtom(\ArrayOf{r_1 \ldots r_N})} =\\
&\;\;\ContentsOfCog{\BTreeAtom\left(\KeyInstance,
    \ArrayAtom(\ArrayOf{\;r_i\;\big |\; r_i \KeyOrderStrict \KeyInstance\;}), 
    \ArrayAtom(\ArrayOf{\;r_i\;\big |\; \KeyInstance \KeyOrder r_i\;})\right )}
\end{split}\end{equation*}
}
\noindent Here $\KeyInstance = \KeyOf{r_{i}}$ for an arbitrary $i$.  
Evaluating the right hand side of the equation recursively and simplifying, we have
\begin{eqnarray*}
&=& \BagOf{\;r_i\;\big |\; r_i \KeyOrderStrict \KeyInstance\;} \uplus 
    \BagOf{\;r_i\;\big |\; \KeyInstance \KeyOrder r_i\;}\\
&=& \BagOf{\;r_i\;\big |\; \left(r_i \KeyOrderStrict \KeyInstance\right) \vee \left(\KeyInstance \KeyOrder r_i\right) \;}\\
&=& \BagOf{r_1 \ldots r_N} \;\;=\;\;\ContentsOfCog{\ArrayAtom(\ArrayOf{r_1 \ldots r_N})}
\end{eqnarray*}
Instances of $\ArrayAtom$ are always structurally correct. The newly created $\BTreeAtom$ instance is structurally correct by construction.  Thus $\CrackTransform$ is correct.
\end{proof}

%%%%%%%%

\begin{lemma}[$\MergeTransform$ is correct]
\label{thm:mergeIsCorrect}
$\MergeTransform$ is both equivalence preserving and structure preserving.
\end{lemma}
\begin{proof}
For any instance $C$ that matches neither of $\MergeTransform$'s cases, correctness follows from \Cref{thm:identityIsCorrect}.
Of the remaining two cases, we first consider
$$C = \ConcatAtom(\ArrayAtom(\ArrayOf{\RecordInstance_1 \ldots \RecordInstance_N}), \ArrayAtom(\ArrayOf{\RecordInstance_{N+1} \ldots \RecordInstance_M}))$$
The proof of equivalence preservation is identical to that of Theorem~\ref{thm:divideIsCorrect} applied in reverse.  In the second case
{\small
$$C = \BTreeAtom(\Wildcard, \ArrayAtom(\ArrayOf{\RecordInstance_1 \ldots \RecordInstance_N}), \ArrayAtom(\ArrayOf{\RecordInstance_{N+1} \ldots \RecordInstance_M}))$$
}
Noting that $\BTreeAtom(\Wildcard, C_1, C_2) \IsLogicallyEquivalentTo \ConcatAtom(C_1, C_2)$ by the definition of logical contents, the proof of equivalence preservation is again identical to that of Theorem~\ref{thm:divideIsCorrect} applied in reverse.
For both cases, structural preservation is given by the fact that $\ArrayAtom$ is always structurally correct.  Thus $\MergeTransform$ is correct.
\end{proof}

%%%%%%%%

\begin{lemma}[$\PivotLeftTransform$ is correct]
\label{thm:pivotLeftIsCorrect}
$\PivotLeftTransform$ is both equivalence preserving and structure preserving.
\end{lemma}
\begin{proof}
For any instance $C$ that matches neither of $\PivotLeftTransform$'s cases, correctness follows from \Cref{thm:identityIsCorrect}.
Of the remaining two cases, we first consider
$$C = \ConcatAtom(C_1, \ConcatAtom(C_2, C_3))$$
Equivalence follows from from associativity of bag union.
{\small
\\[3mm]$\ContentsOfCog{\ConcatAtom(C_1, \ConcatAtom(C_2, C_3))} $
\begin{equation*}\begin{split}
&=\ContentsOfCog{C_1} \uplus \ContentsOfCog{C_2} \uplus \ContentsOfCog{C_3}\\
&=\ContentsOfCog{\ConcatAtom(\ConcatAtom(C_1, C_2), C_3)} 
\end{split}\end{equation*}
}
$\ConcatAtom$ instances are structurally correct if their children are, so the transformed instance is structurally correct if $\alpha(C_1)$, $\alpha(C_2)$, and $\alpha(C_3)$.  
Hence, if the input is structurally correct, then so is the output and the transform is structurally preserving in this case.
The proof of equivalence preservation is identical for the case where 
{\small
$$C = \BTreeAtom(\KeyInstance_1, C_1, \BTreeAtom(\KeyInstance_2, C_2, C_3)) \textbf{ and } \KeyInstance_1 \KeyOrderStrict \KeyInstance_2$$
}
For structural preservation, we additionally need to show:
\textbf{(1)}~$\forall \RecordInstance \in \ContentsOfCog{C_1} : \RecordInstance \KeyOrderStrict \KeyInstance_1$, 
\textbf{(2)}~$\forall \RecordInstance \in \ContentsOfCog{C_2} : \KeyInstance_1 \KeyOrder \RecordInstance$, 
\textbf{(3)}~$\forall \RecordInstance \in \ContentsOfCog{\BTreeAtom(\KeyInstance_1, C_1, C_2)} : \RecordInstance \KeyOrderStrict \KeyInstance_2$, and
\textbf{(4)}~$\forall \RecordInstance \in \ContentsOfCog{C_3} : \KeyInstance_2 \KeyOrder \RecordInstance$ given that $C$ is structurally correct.

Properties (1) and (4) follow trivially from the structural correctness of $C$. 
Property (2) follows from structural correctness of $C$ requiring that 
$\forall \RecordInstance \in (\ContentsOfCog{C_2} \uplus \ContentsOfCog{C_3}) : \KeyInstance_1 \KeyOrder \RecordInstance$
To show property (3), we first use transitivity to show that $\forall \RecordInstance \in \ContentsOfCog{C_1} : \RecordInstance \KeyOrderStrict \KeyInstance_1 \KeyOrderStrict \KeyInstance_2$.  For the remaining records, $\forall \RecordInstance \in \ContentsOfCog{C_2} : \RecordInstance \KeyOrderStrict \KeyInstance_2$ follows trivially from the structural correctness of $C$.
Thus $\PivotLeftTransform$ is correct~\footnote{Note the limit on $\KeyInstance_1 \KeyOrderStrict \KeyInstance_2$, which could be violated with an empty $C_2$.}
\end{proof}

\begin{corrolary}
$\PivotRightTransform$ is correct.
\end{corrolary}

\section{$\LHSMetaTransform$ / $\RHSMetaTransform$ are meta transforms}
\label{apx:lhsrhsAreMeta}

\begin{proof}
We show only the proof for $\LHSMetaTransform$; The proof for $\RHSMetaTransform$ is symmetric.  We first show that $\LHSMetaTransform$ is an endofunctor.  The kind of $\LHSMetaTransform$ is appropriate, so we only need to show that it satisfies the properties of a functor.  
First, we show that $\LHSMetaTransform$ commutes the identity ($\IdentityTransform$).  In other words, for any instance $C$, $\LHSMetaTransform[\IdentityTransform](C) = C$.  In the case where $C = \ConcatAtom(C_1, C_2)$, then 
{\small
$$\LHSMetaTransform[\IdentityTransform](C) = \ConcatAtom(\IdentityTransform(C_1), C_2) = \ConcatAtom(C_1, C_2)$$ 
}
The case where $\TypeOfCogInstance{C} = \BTreeAtom$ is identical, and $\LHSMetaTransform[\TransformInstance]$ is already the identity in all other cases.
Next, we need to show that $\LHSMetaTransform$ distributes over composition.  
That is, for any instance $C$ and transforms $\TransformInstance_1$ and $\TransformInstance_2$ we need that
$$\LHSMetaTransform[\TransformInstance_1 \circ \TransformInstance_2](C) = \left(\LHSMetaTransform[\TransformInstance_1] \circ \LHSMetaTransform[\TransformInstance_2]\right)(C)$$

\noindent If $C = \ConcatAtom(C_1, C_2)$, 
$\LHSMetaTransform[\TransformInstance_1 \circ \TransformInstance_2](C) = \ConcatAtom(C_1', C_2)$, where $C_1' = \TransformInstance_2(\TransformInstance_1(C_1))$.  For the other side of the equation: 
{\footnotesize
\begin{eqnarray*}
\left(\LHSMetaTransform[\TransformInstance_1] \circ \LHSMetaTransform[\TransformInstance_2]\right)(C) 
&=& \LHSMetaTransform[\TransformInstance_2](\LHSMetaTransform[\TransformInstance_1](C))\\
&=& \LHSMetaTransform[\TransformInstance_2](\ConcatAtom(\TransformInstance_1(C_1), C_2)\\
&=& \ConcatAtom(\TransformInstance_2(\TransformInstance_1(C_1)), C_2)
\end{eqnarray*}
}
The case where $\TypeOfCogInstance{C} = \BTreeAtom$ is similar, and the remaining cases follow from $\LHSMetaTransform[\TransformInstance] = \IdentityTransform$ for all other cases.
Thus $\LHSMetaTransform$ is an functor.  
For $\LHSMetaTransform$ to be a meta transform, it remains to show that for any correct transform $\TransformInstance$, $\LHSMetaTransform[\TransformInstance]$ is also correct.
We first consider the case where $C = \ConcatAtom(C_1, C_2)$ and assume that $\TransformInstance(C_1)$ is both equivalence and structure preserving, or equivalently that $\ContentsOfCog{C_1} = \ContentsOfCog{\TransformInstance(C_1)}$ and $\IsStructurallyCorrect{C_1} \implies \IsStructurallyCorrect{\TransformInstance(C_1)}$.  
\begin{eqnarray*}
\ContentsOfCog{\LHSMetaTransform[\TransformInstance](C)}
&=& \ContentsOfCog{\ConcatAtom(\TransformInstance(C_1), C_2)}\\
&=& \ContentsOfCog{\ConcatAtom(C_1, C_2)} = \ContentsOfCog{C}
\end{eqnarray*}
Thus, $\LHSMetaTransform[\TransformInstance]$ is equivalence preserving for this case.  The proof of structure preservation follows a similar pattern
{\footnotesize
\begin{eqnarray*}
\IsStructurallyCorrect{\LHSMetaTransform[\TransformInstance](C)}
&=& \IsStructurallyCorrect{\ConcatAtom(\TransformInstance(C_1), C_2)}\\
&=& \IsStructurallyCorrect{\TransformInstance(C_1)} \wedge \IsStructurallyCorrect{C_2}
\end{eqnarray*}
}
Given $\IsStructurallyCorrect{C} = \IsStructurallyCorrect{C_1} \wedge \IsStructurallyCorrect{C_2}$ and the assumption of $\IsStructurallyCorrect{C_1} \implies \IsStructurallyCorrect{\TransformInstance(C_1)}$, it follows that $\LHSMetaTransform[\TransformInstance]$ is structure preserving for this $C$.
The proof for the case where $C = \BTreeAtom(\KeyInstance, C_1, C_2)$ is similar, but also requires showing that $\forall \RecordInstance \in \ContentsOfCog{\TransformInstance(C_1)} : \RecordInstance \KeyOrderStrict \KeyInstance$ under the assumption that $\forall \RecordInstance \in \ContentsOfCog{C_1} : \RecordInstance \KeyOrderStrict \KeyInstance$.  
This follows from our assumption that $\ContentsOfCog{\TransformInstance(C_1)} = \ContentsOfCog{C_1}$.
The remaining cases of $\LHSMetaTransform$ are covered under \Cref{thm:identityIsCorrect}.  Thus, $\LHSMetaTransform$ is a meta transform.
\end{proof}

\section{Target Updates are Bounded}
\label{apx:boundedTargetUpdates}

\begin{proof}
By recursion over $\TransformInstance$.  The atomic transforms are the base case.
By definition $\IdentityTransform$ is not in the active domain, so we only need to consider seven possible atomic transforms.
For $\SortTransform$ or $\UnsortTransform$ to be in the active domain, $\TypeOfCogInstance{C}$ must be $\ArrayAtom$ or $\SortedAtom$ respectively.  
By the definition of each transform, $\TypeOfCogInstance{C'}$ will be $\SortedAtom$ or $\ArrayAtom$ respectively
By \Cref{thm:hierarchicalAreEnumerable}, the active domain of any $\ArrayAtom$ or $\SortedAtom$ instance is bounded by $|\AtomicTransformSet|$ and by construction, $|\WeightedTargets_C| = |\PolicyDomain_C| \leq |\AtomicTransformSet|$.  
Hence, the total change in the weighted targets for this case is at most $2 \times |\AtomicTransformSet|$.
Following a similar line of reasoning, the weighted targets change by at most $4 \times |\AtomicTransformSet|$ elements as a result of any $\DivideTransform$, $\CrackTransform$, or $\MergeTransform$.
Next consider $C = \ConcatAtom(\ConcatAtom(C_1, C_2), C_3)$, and consequently $C' = \PivotLeftTransform(C) = \ConcatAtom(C_1, \ConcatAtom(C_2, C_3))$.
For each transform of the form $\LHSMetaTransform[\LHSMetaTransform[T]]$ in the active domain of $C$, there will be a corresponding $\LHSMetaTransform[T]$, as $C_1$ is identical in both paths.  
Similar reasoning holds for $C_2$ and $C_3$.  
Because the policy is local, the weighted targets are independent of any $\LHSMetaTransform$ or $\RHSMetaTransform$ meta transforms modifying them.
Thus, at most, the active domain will lose $T$ and $\LHSMetaTransform[T]$ for $T \in \AtomicTransformSet$, and gain $T$ and $\RHSMetaTransform[T]$ for $T \in \AtomicTransformSet$, and the weighted targets will change by no more than $4 \times |\AtomicTransformSet|$ elements.
Similar lines of reasoning hold for the other case of $\PivotLeftTransform$ and for both cases of $\PivotRightTransform$.  
The recursive cases are trivial, since the weighted targets are independent of prefixes in a local policy.
\end{proof}